\documentclass[11pt]{article} 
\usepackage{times}
\usepackage[letterpaper, left=1in, right=1in, top=1in, bottom=1in]{geometry}

\usepackage[utf8]{inputenc} 
\usepackage[T1]{fontenc}    
\usepackage{hyperref}       
\usepackage{url}            
\usepackage{booktabs}       
\usepackage{amsfonts}       
\usepackage{nicefrac}       
\usepackage{microtype}      
\usepackage{xcolor}         
\usepackage{mkolar_definitions}
\usepackage{thmtools, thm-restate}
\usepackage[ruled, vlined]{algorithm2e}
\usepackage{graphicx}
\usepackage{subcaption}
\usepackage{textcomp}

\definecolor{forestgreen}{RGB}{34, 139, 34}

\newcommand{\jobset}{\Jcal}
\newcommand{\canset}{\Ccal}
\newcommand{\binrelcj}{r(c,j)}
\newcommand{\binreljc}{\tilde{r}(j,c)}
\newcommand{\canpre}[2]{f_{#1}(#2)}
\newcommand{\jobpre}[2]{g_{#1}(#2)}
\newcommand{\concanpre}[2]{\tilde{f}_{#1}(#2)}
\newcommand{\conjobpre}[2]{\tilde{g}_{#1}(#2)}
\newcommand{\asycanpre}[2]{\bar{f}_{#1}(#2)}
\newcommand{\asyjobpre}[2]{\bar{g}_{#1}(#2)}
\newcommand{\examvec}{\mathbf{v}}
\newcommand{\examfunc}[1]{v(#1)}
\newcommand{\clickpro}[3]{\PP^{#1}_{{#2},{#3}}}
\newcommand{\swobj}[1]{\textbf{SW}(#1)}
\newcommand{\indiswobj}[2]{\textbf{U}_{#2}(#1)}
\newcommand{\swlowobj}[1]{\underline{\textbf{SW}}(#1)}
\newcommand{\dsmatrix}[2]{M^{#1}_{#2}}
\newcommand{\dsmatrices}[1]{\mathbf{M}^{#1}_{\Ccal}}
\newcommand{\applied}[3]{Y^{#1}_{{#2},{#3}}}
\newcommand{\appliedset}[2]{\Ccal^{#1}_{#2}}
\newcommand{\rank}[2]{rank(#1|#2)}
\newcommand{\rankinlist}[2]{\varphi_{#1}(#2)}
\newcommand{\rankinapplied}[3]{\mathbf{rk}^{#1}_{#2}(#3)}
\newcommand{\caninrank}[2]{\varphi^{-1}_{#1}(#2)}
\newcommand{\ranking}[1]{\sigma(#1)}
\newcommand{\priorset}[2]{A_{#1}({#2})}
\newcommand{\basisvec}[1]{\mathbf{e}_{#1}}
\newcommand{\naivepolicy}{\pi^n}
\newcommand{\optimalpolicy}{\pi^*}
\newcommand{\swname}{\emph{Social-Welfare Ranking}}
\newcommand{\naivename}{\emph{Naive-Relevance Ranking}}
\newcommand{\reciprocalname}{\emph{Reciprocal-Relevance Ranking}}
\newcommand{\ccomplement}{\tilde{c}}

\usepackage{authblk}
\title{Optimizing Rankings for Recommendation in Matching Markets}
\date{}
\author[1]{Yi Su\thanks{ys756@cornell.edu}}
\author[2]{Magd Bayoumi\thanks{mb2363@cornell.edu}}
\author[2]{Thorsten Joachims\thanks{tj@cs.cornell.edu}}
\affil[1]{Department of Statistics and Data Science, Cornell University, Ithaca, NY}
\affil[2]{Department of Computer Science, Cornell University, Ithaca, NY}
\begin{document}

\maketitle

\begin{abstract}

Based on the success of recommender systems in e-commerce, there is growing interest in their use in matching markets (e.g., labor). While this holds potential for improving market fluidity and fairness, we show in this paper that naively applying existing recommender systems to matching markets is sub-optimal.
Considering the standard process where candidates apply and then get evaluated by employers, we present a new recommendation framework to model this interaction mechanism and propose efficient algorithms for computing personalized rankings in this setting.
We show that the optimal rankings need to not only account for the potentially divergent preferences of candidates and employers, but they also need to account for capacity constraints. This makes conventional ranking systems that merely rank by some local score (e.g., one-sided or reciprocal relevance) highly sub-optimal --- not only for an individual user, but also for societal goals (e.g., low unemployment). 
To address this shortcoming, we propose the first method for jointly optimizing the rankings for all candidates in the market to explicitly maximize social welfare. 
In addition to the theoretical derivation, we evaluate the method both on simulated environments and on data from a real-world networking-recommendation system that we built and fielded at a large computer science conference. 
\end{abstract}

\section{Introduction}

Most search and recommender systems rely on rankings as the prevalent way of presenting results to the users. By ranking the most promising items first, they aim to focus the user's attention on a tractably small consideration set especially when the number of items is large. For conventional applications of ranking systems in online retail or media streaming, it has long been understood that ranking items based on their probability of relevance (alternatively, purchase, stream or click) to the user provides maximal utility under standard assumptions. Sorting by probability of relevance is commonly called the Probability Ranking Principle (PRP) \cite{robertson1977probability}, and it implies that the optimal ranking for any user depends only on his or her preferences. 

The PRP is no longer optimal, however, for a growing range of new online platforms that mediate matching markets like job search, college admission, dating, or social recommendations \cite{ashlagi2019assortment, liu2020competing, mladenov2020optimizing, tu2014online}. In these matching markets, the online platform acts as a mediator between both sides of the market, connecting job candidates with employers or colleges with applicants. The following two key properties of matching markets make the ranking problem substantially more complex than ranking by each individual's probability of relevance alone. First, a successful match requires that both sides of a match agree on mutual relevance \cite{pizzato2013recommending, li2012meet}. Second, both sides have constraints on how many options they can evaluate. Under these conditions, the optimal ranking for any one participant depends on the rankings and relevances of the other participants in the market, which can make PRP ranking highly suboptimal.

To illustrate the source of these dependencies, consider a job matching problem where many of the job candidates crowd to a small set of employers with high name-recognition. The popular employers will receive many applications --- more than they can carefully evaluate and many without the desired qualifications. At the same time, less well-known employers may not get discovered by relevant candidates. The result is that many candidates never hear back from the popular employers and many employers are unable to fill their positions. This is not only bad for individual employers and candidates, but also a poor outcome in terms of the social welfare that the system provides (e.g. minimizing unemployment and unfilled jobs). To remedy this problem, a recommender system for job matching should elicit preferences and qualifications from both the candidates (\emph{e.g.,} west-coast location, python) and the employers (\emph{e.g.,} 5 years experience, python) and help discover mutually beneficial matches that avoid crowding. 

The main contribution for this paper is fourfold: (1) First, the paper is the first to formalize the problem of \emph{ranking in two-sided matching markets} under the apply-accept interaction protocol that is used in many real-world settings.
(2) Second, we present the first recommender system that is able to jointly optimize personalized rankings for all candidates to maximize social welfare in two-sided matching markets. Note that our recommendation approach leaves all participants with autonomy over their actions and eventual decisions, which makes it different from conventional matching procedures \cite{gale1962college, masarani1989existence}, where all participants have to submit to a centralized matching procedure \cite{roth1999redesign} that determines the matching outcome. Instead, our approach supports each participant in discovering the most promising options to manually evaluate, optimizing how the participants spend their scarce resource of attention among an often overwhelming set of options.  (3) Third, we discuss additional strategic behavior (e.g., adoption and retention) and fairness issues for rankings under this framework in Section~\ref{sec:fair}, which opens an important area for future work. (4) Finally, we built a real-world networking-recommendation system that we fielded at a large computer science conference
. We empirically validate our method on this real-world dataset and will publish this data as a benchmark to enable future research.

\section{Related Work}
As \emph{multi-sided market platforms} have emerged as a popular business model, search and recommender systems have taken a key role in mediating their interactions. These systems need to balance the interests of various stakeholders, such as the users, the suppliers, and the platform itself.
There is a large body of recent work on how to specify each stakeholder's objective \cite{mehrotra2020advances}, their interplay \cite{mehrotra2018towards, wang2020fairness}, and the design of efficient joint optimization objectives for recommendation in multi-sided marketplaces \cite{rodriguez2012multiple, mehrotra2020bandit, zhang2014fairness,svore2011learning}. Important objectives include diversity \cite{clarke2008novelty, radlinski2009redundancy}, novelty \cite{vargas2011rank, ribeiro2012pareto} and fairness \cite{singh2018fairness, zafar2017fairness, biega2018equity, beutel2019fairness}. Most of the work relies on techniques from multi-objective optimization, which include finding the Pareto front \cite{parisi2014policy}, using aggregation functions that reduce to a single objective \cite{kim2006adaptive, zhang2014fairness, mehrotra2020bandit}, or including some objectives as constraints \cite{rodriguez2012multiple, singh2019policy, singh2018fairness}. Unlike most of these works, we directly model the interaction process in two-sided matching markets and aim to maximize the overall social welfare, instead of manually trading-off the interest for different stakeholders. 

Our problem setting is a form of \emph{reciprocal recommendation} \cite{pizzato2013recommending, li2012meet}, which considers matching problems where both sides have preferences like online dating \cite{alanazi2013people, tu2014online, xia2015reciprocal}, job seeking \cite{almalis2014content, yu2011reciprocal, mine2013reciprocal}, and social media \cite{zhang2011intrank}. Hence, most reciprocal recommendation systems rank by a reciprocal score that combines the preferences from both sides, like the harmonic mean \cite{pizzato2010learning}, sum of similarities \cite{almalis2014content}, product operator \cite{ting2016transfer}, and community-level matching \cite{alsaleh2011improving}. Among these works, \cite{tu2014online} is the closest to ours. They propose a two-sided matching framework for dating recommendation that maximizes the total number of reciprocated messages, while avoiding to overburden each user by keeping the expected number of messages received/sent bounded. A key simplifying assumption in \cite{tu2014online} is that the system provides a predetermined number of recommendations to each user, and that all these recommendations will be exhaustively evaluated by the user. Similar assumptions are also made by recent work on assortment planning \cite{ashlagi2019assortment} for matching problems. We consider the  realistic case where each user is shown a ranking, and there is uncertainty in how far each user goes down. 

Our work also builds upon the economic and social science literature on {\em matching markets} \cite{gale1962college, masarani1989existence}. Here, the core problem lies in the design of matching procedures. In particular, stable-matching algorithms \cite{gale1962college} take a ranked preference list from both sides of the market as input and produce stable one-to-one or many-to-one matchings, while maximum weight matching \cite{duan2014linear} considers weighted preference graphs. In contrast, our work helps the users {\em form} their preferences despite an intractably large set of possible options. In particular, we do not force the final matching step and leave all participants with autonomy over their actions and eventual decisions. This makes it fundamentally different from conventional matching procedures.

\section{A Framework for Ranking in Matching Market}
\label{sec:framework}

We begin by introducing our new framework for ranking in matching markets. For simplicity and concreteness, we use a \emph{job recommendation platform} as a running example. However, the framework is general and can be adapted to other domains where (1) the success of a matching relies on relevance from both sides and (2) both sides have limited attention.

We consider the following \emph{sequential} and \emph{asymmetric} interaction protocol. First, \emph{job candidates} (proactive side) browse through their ranked lists of jobs and apply to the jobs that they find relevant. After this, \emph{employers} (reactive side) browse through rankings of their applicants and respond to the applicants they find relevant (e.g., invite to interview). We call this a match, and the goal of the ranking system is to design rankings for each candidate to maximize the total number of expected matches in the market. 

We now formalize this interaction model, and then present algorithms for jointly optimizing the rankings for all candidates in Section~\ref{method}. The toy example in Figure~\ref{fig:example} serves as an intuitive guide to the model we develop in the following. 
We denote the set of candidates in the market as $\canset$, and the set of employers as $\jobset$ with $|\canset|< \infty$ and $|\jobset|<\infty$. We use $\binrelcj \in \{0,1\}$ to denote the binary relevance of employer $j$ to candidate $c$ (\emph{i.e.,} $c$ wants to apply to $j$). Similarly, $\binreljc \in \{0,1\}$ denotes the relevance of candidate $c$ to employer $j$ (i.e., $j$ wants to interview $c$). However, the ranking system does not know the precise relevances, but it merely has access to relevance probabilities $$\canpre{c}{j}:=\PP(\binrelcj=1) \:\:\:\text{  and  }\:\:\: \jobpre{j}{c}:=\PP(\binreljc=1).$$ 
A large body of literature exists on how to estimate probabilities of relevance, and for the purposes of this paper we assume that accurate and unbiased estimates exist.

\begin{figure}[t!]
    \centering
        \includegraphics[width=\linewidth]{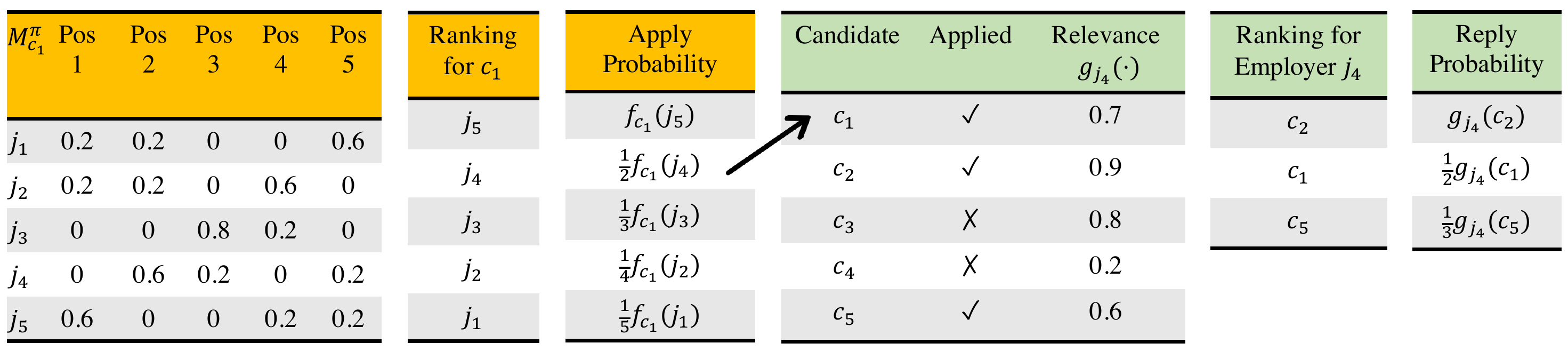}

    \caption{A two-sided matching market with $|\canset|=5$ candidates, $|\jobset|=5$ employers, and examination function $v(x)=1/x$. The yellow tables show the interaction process for a particular candidate $c_1$, while the green tables are for a particular employer $j_4$ (other candidates and employers follow an analogous process). From left to right, the system first computes a personalized stochastic ranking $\dsmatrix{\pi}{c_1}$ (the $5\times 5$ matrix) for $c_1$. Then it samples a particular ranking $(j_5, j_4, j_3, j_2, j_1)$, presents it to $c_1$, and $c_1$ applies to the employers according to the position-based model. For the employer side, one realization of the process shows $\{c_1, c_2, c_5\}$ applying to $j_4$. The system thus presents ranking $(c_2,c_1,c_5)$ to $j_4$ sorted by $g_{j_4}(\cdot)$. Finally, $j_4$ replies back also following a position-based model.}
    \label{fig:example}
\end{figure}

\subsection{Candidates Act (Proactive)}
Like in many real-world markets, candidates act first in our model. Our recommendation system is characterized by a \emph{contextualized stochastic} ranking policy $\pi: \canset \to \Delta_{\Sigma_{|\jobset|}}$, which maps each candidate to a probability distribution over rankings of the employers $\jobset$. The space of all possible rankings is denoted by $\Sigma_{|\jobset|}$. For each candidate $c$, the system samples a particular deterministic ranking $\sigma(c)$ from $\pi(\cdot|c)$, \emph{i.e.,} $\ranking{c}\sim \pi(\cdot|c)$ and presents it to the candidate $c$. 
 
After each candidate $c$ receives a ranking of employers $\ranking{c}$, they act independently by going down the ranking and applying to an employer if it is relevant. We use the position-based model (PBM) \cite{richardson2007predicting,chuklin2015click,joachims2017unbiased} from the information retrieval literature to model this application process. Given a ranking $\ranking{c}$, the PBM models the probability of candidate $c$ applying to job $j$ when given ranking $\ranking{c}$ as the product of the relevance probability $\canpre{c}{j}$ and the examination probability $\PP(e(j)=1|\ranking{c})$:
  \begin{equation}
  \PP(c \text{ applies to } j | \ranking{c}) = \canpre{c}{j} \cdot \PP(e(j)=1|\ranking{c})
 \end{equation}
The examination probability models how likely a candidate will discover a job that is placed in position $k$. In the PBM, this examination probability depends only on the rank $\rank{j}{\ranking{c}}$ of employer $j$ under ranking $\ranking{c}$. Therefore, we can rewrite the examination probability as:
 \begin{equation}
 	 \PP(e(j)=1|\ranking{c}) = \examfunc{\rank{j}{\ranking{c}}}.
 \end{equation}
$v$ is an application-dependent function. Common choices include $\examfunc{x}=1/x$ \cite{joachims2017unbiased} and $\examfunc{x}=\frac{1}{\log(1+x)}$ \cite{jarvelin2002cumulated}, or an application-dependent $v$ can be estimated directly \cite{agarwal2019estimating}.
Combining this, the probability $\clickpro{\pi}{c}{j}$ of candidate $c$ applying to employer $j$ under ranking policy $\pi$ can be expressed as:
 \begin{equation}
 \label{eq:capplies}
 \!\!\!\clickpro{\pi}{c}{j}:=\PP(c \text{ applies to } j|\pi) = \!\!\!\!\!\!\!\!\! \sum_{\ranking{c}\in\Sigma_{|\Jcal|}}\!\!\!\!\!\!\!\!\pi(\ranking{c}|c) \canpre{c}{j}v\big(\rank{j}{\ranking{c}}\big).
 \end{equation}
This expression may appear intractable at first glance, since it involves a sum over the exponential number of possible rankings.
However, note that the relevance probability $\canpre{c}{j}$ does not depend on the rank and that the examination probability $v\big(\rank{j}{\ranking{c}}\big)$ only depends on rank, such that $\clickpro{\pi}{c}{j}$ can be expressed equivalently in terms of the marginal probabilities $\PP(\rank{j}{\ranking{c}})=k|\pi)$ of employer $j$ being placed at rank $k$ under policy $\pi$. Therefore, $\clickpro{\pi}{c}{j}$ could be re-written as:
  \begin{equation}
   \clickpro{\pi}{c}{j} = \canpre{c}{j}\big(\sum_{k=1}^{|\jobset|}\PP(\rank{j}{\ranking{c}}=k|\pi)\examfunc{k}\big)= \canpre{c}{j}\big(\sum_{k=1}^{|\jobset|}\dsmatrix{\pi}{c}(j,k)\examfunc{k}\big)           
 \end{equation}
with $\dsmatrix{\pi}{c}$ denotes the doubly stochastic matrix for candidate $c$ where the $(j,k)$-th position equals $\PP(\rank{j}{\ranking{c}})=k|\pi)$. This enables us to use only $|\jobset|^2$ parameters when representing ranking policy $\pi(\cdot|c)$ for candidate $c$ via $\dsmatrix{\pi}{c}$, since all stochastic ranking policies with the same $\dsmatrix{\pi}{c}$ are equivalent. Meanwhile, we use $\dsmatrices{\pi}$ as the concatenation of $\dsmatrix{\pi}{c}$ $ \forall c\in\canset$ to denote the set of doubly stochastic matrices for all candidates. In the following, this will allow us to optimize in the space of doubly stochastic matrices instead of in the exponentially sized space of rankings. Once an optimal $\dsmatrix{\pi}{c}$ is found, we can use the Birkhoff-von Neumann (BvN) decomposition to efficiently find a stochastic ranking policy that corresponds to $\dsmatrix{\pi}{c}$ in polynomial time and thus sample a deterministic ranking and present to the candidate \cite{birkhoff1940lattice, singh2018fairness}.

To further compact the notation, we write the examination probabilities\footnote{Here, we assume user-homogeneous examination. It could be easily extended into user-heterogenous case.} as a vector $\examvec\in \RR_{+}^{|\jobset|}$ with $\examvec[k] = \examfunc{k}$ for $k\in\{1,2,3,\cdots, |\jobset|\}$, and we use $\basisvec{j}$ as the standard basis vector in $|\jobset|$-dimensional vector space, with a 1 in the $j$-th position and 0's elsewhere. Now the probability of candidate $c$ applying to employer $j$ can be written as: 
 \begin{equation}
 \label{equ:cus_click}
 	 \clickpro{\pi}{c}{j} = \canpre{c}{j}\basisvec{j}^T\dsmatrix{\pi}{c} \examvec.
 \end{equation}

\subsection{Employers Act (Reactive)}

After collecting all the applications, it is the employers' turn to reply back.
To be specific, the ranking system will give each employer $j$ a ranked list of candidates \emph{who applied to them}, sorted by their employer-specific probability of relevance $\jobpre{j}{\cdot}$. We denote the set of candidates who applied to employer $j$ as $\appliedset{\pi}{j} \subseteq \canset$, and it is important to point out that $\appliedset{\pi}{j}$ is stochastic and the randomness comes from the candidates' application outcome under policy $\pi$.
Then each employer responds back to the candidates $c\in \appliedset{\pi}{j}$ following a similar PBM model as the one we introduced for the candidates\footnote{In general, the examination function $\examfunc{\cdot}$ could be different for candidates and employers. Here we keep it the same for simplicity of notation.}. In particular, the probability that employer $j$ replies back to candidate $c$ (e.g., invites to interview)  depends on the relevance probability $\jobpre{j}{c}$ and the examination probability. Note that the examination probability for any candidate $c$ who applied to $j$ depends on who else is in $\appliedset{\pi}{j}$, inducing a distribution over ranks. Let $\rankinapplied{\pi}{j}{c}$ denote the rank of candidate $c$ in a particular $\Ccal^\pi_j$ given $c$ applied to $j$, ranked by $\jobpre{j}{c}$. We can then write the probability of a reply as follows:
\begin{equation}
\label{equ: job_click}
\clickpro{\pi}{j}{c}:=\PP(j \text{ replies to } c|c \text{ applied to }j, \pi) =  \jobpre{j}{c}\EE\big[v\big(\rankinapplied{\pi}{j}{c}\big)\big]
 \end{equation}
The expectation is over the randomness in $\Ccal^\pi_j$, which is induced by the application process on the side of the candidates. 

\subsection{Utility and Social Welfare}

We now formulate the objective of the ranking system, which is based on the notion of successful matches. A match for a candidate/employer pair $(c,j)$ is successful, if candidate $c$ manages to discover a relevant $j$ and thus applies, and employer $j$ also finds $c$ relevant and manages to discover this candidate among the applicants. We use $\applied{\pi}{c}{j}\in\{0,1\}$ to denote whether a match is successful. We define the utility $\indiswobj{\pi}{c}$ that the ranking system provides to candidate $c$ when using ranking policy $\pi$ as the expected number of matches (e.g., interviews) that $c$ receives.
\begin{equation}
 \begin{aligned}
     \label{indi_SW}
    \indiswobj{\pi}{c}&=\sum_{j\in\jobset}\PP(\applied{\pi}{c}{j}=1)
	= \sum_{j\in\jobset} \clickpro{\pi}{c}{j} \clickpro{\pi}{j}{c} =\sum_{j\in\jobset} \big(\canpre{c}{j}\basisvec{j}^T\dsmatrix{\pi}{c} \examvec\big)\jobpre{j}{c}\EE\big[v\big(\rankinapplied{\pi}{j}{c}\big)\big]\\
	&=\sum_{j\in\jobset}\bigg(\canpre{c}{j}\jobpre{j}{c}\EE\big[v\big(\rankinapplied{\pi}{j}{c}\big)\big]\bigg)\basisvec{j}^T\dsmatrix{\pi}{c} \examvec
 \end{aligned}
 \end{equation}
After rearranging the terms, we can see that the utility of ranking policy $\pi$ for candidate $c$ is analogous to traditional ranking measures like DCG \cite{jarvelin2002cumulated}. The key difference is that the utility for each ranked item not only depends on the candidates' own relevance, but also on those of the employers and potentially all other candidates, since they affect the position in the employers' rankings. One can also define a similar utility $\indiswobj{\pi}{j}$ for each employer.

Given this \emph{interdependent} nature of the individual utilities, and the societal role that many of these matching markets play, we focus on \emph{social welfare} as the overall objective of the ranking system. We select the most straightforward definition of social welfare $\swobj{\pi}$, which is the expected number of matches in the market that ranking policy $\pi$ produces.
Note that this is equivalent to maximizing the sum of candidate utilities, or equivalently the sum of employer utilities.
\begin{equation}
\label{SW}
	\swobj{\pi}=\sum_{c\in\canset} \indiswobj{\pi}{c} = \sum_{j\in\jobset} \indiswobj{\pi}{j}
\end{equation}
Maximizing the total number of matches is arguably a sensible proxy for societal goals (e.g., minimize unemployment), although one could further refine this objective (e.g. diminishing returns for individuals \cite{nemhauser1978analysis}). In particular, an important additional consideration is ensuring the fairness of such a ranking system, which we futher discuss in Section~\ref{sec:fair} and Appendix~\ref{sec:tradeoff}.

In summary, we define the problem of \emph{Social-Welfare Maximization of Rankings for Two-Sided Matching Markets} that we consider in this paper as follows.

\begin{definition} 
Given a two-sided matching market with a proactive side $\canset$ and a reactive side $\jobset$, along with
two-sided relevance probability $\canpre{c}{j}, \jobpre{j}{c}$ $\forall c\in \canset, j\in\jobset$, 
the goal is to design a stochastic ranking policy $\pi$ that maximizes the expected number of matches in the market.
\end{definition}

While an abstraction of the real-world, our model captures the key interactions and trade-off's faced in many modern two-sided matching markets, where (1) the number of available options is often very large, (2) both sides of the market do not have the resources to evaluate all options exhaustively, and (3) a successful transaction depends on a two-sided relevance match.  

It is worth pointing out that our ranking problem is fundamentally different from stable matching problems considered in the economics literature \cite{gale1962college, masarani1989existence}, and that stable matching procedures cannot be used to solve our optimization problem. We provide further discussion in Appendix~\ref{app:stable}




\section{Optimizing Rankings in Matching Markets} 
\label{method}
 
We now explore algorithms for computing ranking policies $\pi$ that maximize social welfare. To motivate the need for these new algorithms, we start by theoretically quantifying the sub-optimality of conventional PRP policies that rank based on one-sided relevances. 
Denote with $\naivepolicy$ the naive relevance-based policy such that for each candidate $c$, $\naivepolicy(\cdot|c)$ sorts by $\canpre{c}{j}$ and consequently recommends job $j$ at position $\rankinlist{c}{j}$ with probability 1, where $\rankinlist{c}{j}$ denotes the rank of $j$ in $\jobset$ when ranked by $\canpre{c}{j}$. Equivalently, consider its equivalent representation using a doubly-stochastic matrix $\dsmatrices{\naivepolicy}$, which we will compare against the social welfare optimal ranking $\dsmatrices{\optimalpolicy}$.

\begin{restatable}{theorem}{subopt}
\label{thm: naive_sub}
There exist two-sided matching markets over $|\canset| = |\jobset| = n$ with examination models $\examfunc{\cdot}$ and relevance probabilities $\canpre{c}{j}$ and $\jobpre{j}{c}$, such that the gap in social welfare between the optimal ranking $\dsmatrices{\optimalpolicy}$ and the naive one-side relevance-based ranking $\dsmatrices{\naivepolicy}$ is larger than $\Theta(n)$. 
\end{restatable}
The proof is in Appendix~\ref{app: proof}. The key idea is to construct an instance where there is a popular employer $j^{*}$ that has high probability of relevance for all candidates. The naive relevance-based ranking will rank $j^{*}$ at the top position for all candidates and create crowding, while the optimal policy will also consider less crowded employers that may have negligibly smaller relevance probability. 

 \subsection{A Tractable Optimization Objective}

Given the sub-optimality of naive rankings, we now develop algorithms that directly optimize social welfare. We first show that the original objective is intractable and hard to optimize directly. To address this, a lower bound is derived and we propose efficient algorithms for solving it. 

Before diving into the analysis, we will introduce some notations that will be used in this section. 
Analogous to $\rankinlist{c}{j}$, we define $\rankinlist{j}{c}\in\{1,2,3,\cdots, |\canset|\}$ as the rank of candidate $c$ when ranking all candidates in $\canset$ by employer $j$'s relevance probability $\jobpre{j}{c}$. We break ties arbitrarily. 
Inversely, we define $\caninrank{j}{s}:=\{c'\in\canset|\rankinlist{j}{c'}=s\}$ as the candidate listed in rank $s$ among $\canset$ when ranked by $\jobpre{j}{c}$.
Correspondingly, we define the priority set for employer $j$ w.r.t. candidate $c$ as $\priorset{j}{c}:=\{c'\in\canset|\rankinlist{j}{c'}< \rankinlist{j}{c}\}$, which includes the candidates who have higher relevance probability to employer $j$ than candidate $c$, measured by $\jobpre{j}{c}$. Based on this, let $F_{(j,c)}^l$ with $l\leq |\priorset{j}{c}|$ be the set of all subsets of $l$ items that can be selected from $\priorset{j}{c}$: $F_{(j,c)}^l:= \{B\subseteq \priorset{j}{c}| |B| = l\}$.

Given examination function $\examfunc{\cdot}$, and two-sided relevance probabilities $\canpre{c}{j}$ and $\jobpre{j}{c}$, note that our \emph{social welfare} objective can be written as:
\begin{equation}
\begin{aligned}
\label{SW_detail}
\swobj{\pi}
	&= \sum_{c\in\canset}\sum_{j\in\jobset} \clickpro{\pi}{c}{j} \clickpro{\pi}{j}{c}= \sum_{c\in\canset}\sum_{j\in\jobset} \underbrace{\big(\canpre{c}{j}\basisvec{j}^T\dsmatrix{\pi}{c} \examvec\big)}_{\clickpro{\pi}{c}{j}}\underbrace{\jobpre{j}{c}\EE\big[v\big(\rankinapplied{\pi}{j}{c}\big)\big]}_{\clickpro{\pi}{j}{}}
\end{aligned}
\end{equation} 
The probability $\clickpro{\pi}{c}{j}$ of candidate $c$ applying to employer $j$ is tractably linear in the policy $\pi$ (or its corresponding doubly stochastic matrices $\dsmatrices{\pi}$). The difficulty lies in analyzing the term $\rankinapplied{\pi}{j}{c}$, which is a random variable that depends on whether other candidates $c\in \canset$ apply to employer $j$ and the relative relevance among them. The following lemma provides the exact distribution of $\rankinapplied{\pi}{j}{c}$, with the parameters depending on the ranking policy $\pi$.

\begin{restatable}{lemma}{pdf}
\label{distribution}
	 For any fixed $j\in\jobset$, $c\in\canset$, under ranking policy $\pi$, the rank $\rankinapplied{\pi}{j}{c}$ of candidate $c$ in the set of candidates that applied to employer $j$ is a random variable with $\rankinapplied{\pi}{j}{c}\sim 1+ X^{\pi}_{j,c}$, where $X^{\pi}_{j,c}$ is a Poisson Binomial random variable with parameters 
	 $ \big[\clickpro{\pi}{\caninrank{j}{1}}{j}, \clickpro{\pi}{\caninrank{j}{2}}{j}, \cdots, \clickpro{\pi}{\caninrank{j}{\rankinlist{j}{c}-1}}{j}\big]$. 
	 Each element corresponds to the probability of candidate $c'\in \priorset{j}{c}$ applying to employer $j$. 
	 For $k\in \{1,2,\cdots \rankinlist{j}{c}\}$, we have:
	 \begin{equation}
\label{eq:dist}
\PP(\rankinapplied{\pi}{j}{c}= k) = \sum_{U \in F^{k-1}_{(j,c)}} \prod_{s \in U} \clickpro{\pi}{s}{j}  \prod_{r \in A_j(c)\backslash U}(1-\clickpro{\pi}{r}{j})
\end{equation}
\end{restatable}

The proof is provided in Appendix~\ref{app: proof}. Given Lemma~\ref{distribution}, it is easy to see that the PMF of the distribution of $\rankinapplied{\pi}{j}{c}$ involves $n!/((n-l)!l!)$ terms of summation (with $n=|\priorset{j}{c}|$ and $l\in\{1,2,\cdots, \rankinlist{j}{c}-1\}$), which poses an extensive computational burden. While a recursive formula exists in the literature \cite{shah1973distribution}, the complicated form w.r.t.\ the $\clickpro{\pi}{c'}{j}$'s poses significant challenges when we treat these probabilities as function of $\pi$ (or its $\dsmatrices{\pi}$). To avoid this, we instead optimize the following lower bound on the original objective (Equation~\ref{SW_detail}), which only requires a mild condition on the examination function $\examfunc{\cdot}$ and is much easier to compute and optimize. 

\begin{restatable}{theorem}{uti}
If the examination function $\examfunc{\cdot}$ is convex, the following expression lower bounds the social-welfare objective of the stochastic ranking policy $\pi$:
\begin{equation}
\label{eq:swlower}
	\swlowobj{\pi}:= \sum_{c\in\canset}\sum_{j\in\jobset}\! \canpre{c}{j}\jobpre{j}{c}v\big(1+\!\!\!\!\!\!\!\sum_{c' \in \priorset{j}{c}}\!\!\!\!\!\canpre{c'}{j}\basisvec{j}^T \dsmatrix{\pi}{c'} \examvec\big)\basisvec{j}^T \dsmatrix{\pi}{c} \examvec
\end{equation}
\end{restatable}

It is worth noting that convexity of $v$ is not restrictive, since most of the commonly used examination models like $\examfunc{x}=1/x$ and $\examfunc{x}=1/\log_2(1+x)$ satisfy this condition. 

Now we are ready to formulate the main optimization problem, which optimizes the tractable \emph{social-welfare-aware} objective $\swlowobj{\pi}$ over the set of doubly stochastic matrices, one for each candidate:
\begin{equation}
\label{eq:lower_opt}
\begin{aligned}
& \underset{\dsmatrices{\pi}:=\{\dsmatrix{\pi}{c}\}_{c\in\Ccal}}{\text{maximize}}
& &  \!\!\!\!\!\!\sum_{c\in\canset}\sum_{j\in\jobset} \canpre{c}{j}\jobpre{j}{c}v\big(1+\!\!\!\!\!\!\sum_{c' \in \priorset{j}{c}}\!\!\!\!\!\!\canpre{c'}{j}\basisvec{j}^T \dsmatrix{\pi}{c'} \examvec\big)\basisvec{j}^T \dsmatrix{\pi}{c} \examvec\\
& \hspace{0.7cm}\text{s. t.}
& & \!\!\!\!\!\!\textbf{1}^T \dsmatrix{\pi}{c} = \textbf{1}^T, \dsmatrix{\pi}{c}\textbf{1}=\textbf{1}, \forall c\in\canset   \; 
\end{aligned}
\end{equation}

Multiple approaches can be used to optimize Equation~\eqref{eq:lower_opt}. One option is projected gradient descent, where the projection of any positive matrix into the set of doubly stochastic matrices can be computed by the Sinkhorn-Knopp Algorithm \cite{sinkhorn1967concerning, wang2010learning}, which is known to minimize the KL divergence of any nonnegative matrix to the Birkhoff Polytope \cite{wang2010learning}. 
An alternative is conditional gradient descent. Since the set of doubly stochastic matrices is convex, we can utilize any convex optimization solver to find the descent direction. In our experiments, we use the Frank-Wolfe approach. Both algorithms are provided in Appendix~\ref{app:alg}. 
For specific examination functions $\examfunc{\cdot}$, more specialized algorithms exists. For example, for $\examfunc{x}=1/x$ the optimization problem in Equation~\eqref{eq:lower_opt} becomes a fractional program that can be optimized with an iterative convex-concave procedure.

\section{Experiments}
\label{sec:exp}
In this section, we empirically evaluate several key properties of our approach. We first present experiments on synthetic data which allows us to vary the properties of the two-sided markets to explore the robustness of the method. In addition, we also assess our method on a real-world data for external validity, including a benchmark dataset from a dating platform, and a new dataset from a virtual-conference networking-recommendation system we built.

\subsection{Analysis on Synthetic Data}

To examine how our method performs in comparison to baselines over a range of matching markets with different characteristics, we create synthetic datasets as follows. We construct matching markets with $n$ employers and $1.5 n$ job candidates to avoid unrealistic symmetry. In the simplest case, we generate relevance probabilities $\asycanpre{c}{j}$ and $\asyjobpre{j}{c}$ through independent and uniform draws from $[0,1]$. We refer to this as the {\em random} setting, but also consider more structural preferences. One type of structure is {\em crowding} on some employers and candidates. To create a setting with crowding, we rank employers and candidates in arbitrary order and name them $j_1, j_2, \cdots, j_{|\jobset|}$ and $c_1, c_2, \cdots, c_{|\canset|}$. For the candidate-side relevance, we linearly interpolate the relevance probability in $[0,1]$ and define $\concanpre{c}{j_i}=1-\frac{i-1}{|\jobset|-1}$ identically for all $c$ so that the relevances of a fixed employer are the same to all candidates.
For the employer-side relevance, we similarly take $\conjobpre{j}{c_i}=1-\frac{i-1}{|\canset|-1}$ for all $j$. To adjust the level of crowding, we take the convex combination of the random setting and the fully crowded setting with parameter $\lambda$: $\canpre{c}{j}:=(1-\lambda) \asycanpre{c}{j}+\lambda \concanpre{c}{j}$ and $\jobpre{c}{j}:=(1-\lambda) \asyjobpre{c}{j}+\lambda \conjobpre{c}{j}$. If not mentioned otherwise, we use $\lambda=0.5$ in the experiments, as well as a market size of $n=100$ and the examination function $\examfunc{x}=1/x$ for all candidates and employers.
We measure the quality for various ranking policies by the original social welfare objective from Equation~\eqref{SW_detail}, which for each run we estimate using 10000 Monte Carlo samples following the generative process described in Section~\ref{sec:framework}, and average over 10 runs.

We compare our \swname, optimized using Algorithm~\ref{alg:1}, against the following baselines. \naivename\ ranks employers by the one-sided relevance $\canpre{c}{j}$ for each candidate $c$. \reciprocalname\ ranks by the reciprocal relevance probability $\canpre{c}{j}\jobpre{j}{c}$ for each candidate $c$, representing a heuristic objective that accounts for the reciprocal nature of a match while ignoring dependencies between candidates. 
\begin{figure*}[!ht]
  \centering
\includegraphics[width=\linewidth]{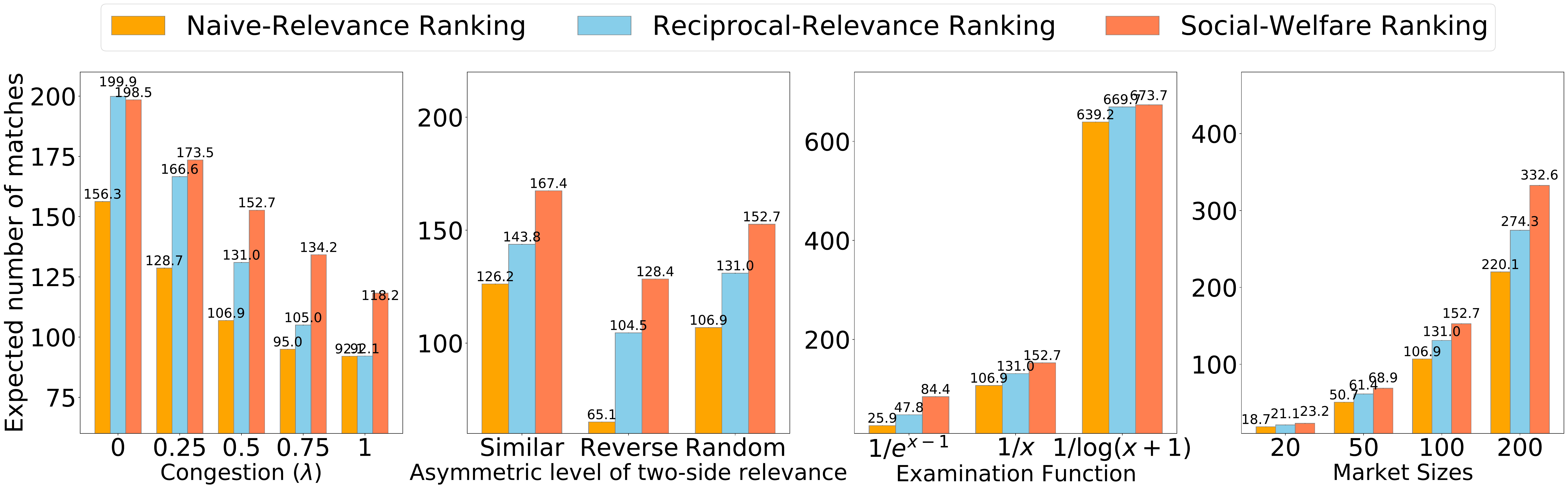}
  \caption{Results on synthetic dataset. The parameters are ($\lambda=0.5$, random, $\examfunc{x}=1/x$, $n=100$), unless stated otherwise. The standard errors are on the order of $1e-2$ and invisible in the graph.}
  \label{fig:combine}
\end{figure*}
\paragraph{How do the methods perform for different levels of crowding?} 
With crowding we refer to a situation where some employers become overloaded with applications while others get unnoticed. The leftmost graph in  Figure~\ref{fig:combine} shows the expected number of matches for different levels $\lambda$ of crowding. Especially for high levels of crowding, the \swname\ performs substantially better than the baselines. While the \naivename\ performs poor for all levels of crowding, the \reciprocalname\ performs equivalent to the \swname\ when the preferences are fully random and there is no crowding. However, the \reciprocalname\ fails to account for collisions in high-crowding settings, where it does no better than the \naivename.

\paragraph{How do the methods perform when there is structure in the relevance probabilities?} 
Most real-world problems will contain some structure in the relevance probabilities $\canpre{c}{j}$ and $\jobpre{j}{c}$. We now explore the two complementary cases where the candidate and employer relevance probabilities are either similar to each other or the reverse of each other. To construct \emph{similar} two-sided relevances, we take $\asyjobpre{j}{c}=\min\{\max\{\asycanpre{c}{j}+e,0\},1\}$ with $e\sim\mathcal{N}(0,0.2)$. For the \emph{reversed} two-sided relevances, we take $\asyjobpre{j}{c}=\min\{\max\{1-\asycanpre{c}{j}+e,0\},1\}$ with $e\sim\mathcal{N}(0,0.2)$. For the \emph{random} two-sided relevances, we use the construction already introduced above. For all settings, we use a crowding level of $\lambda=0.5$. The results are shown in the second plot of Figure~\ref{fig:combine}. As the level of asymmetry moves from reverse to similar, the expected number of matches for all methods increases as expected. Moreover, for all relevance structures, \swname\ consistently achieves substantially higher social welfare than the baseline methods.
\paragraph{How does the examination function influence the relative performance?} The examination function $\examfunc{\cdot}$ models how many results people are able or willing to browse. A steep drop-off in examination probability, like $v(x)=1/e^{x-1}$, means that they are likely to only evaluate the top few results. A flat examination function, such as $\examfunc{x}=\frac{1}{\log(1+x)}$, means that they are likely to go further down. The third plot in Figure~\ref{fig:combine} shows the expected number of matches as we change the examination function. Unsurprisingly, a flatter examination function leads to more matches and little difference between the methods, since results are likely to be discovered no matter where they are placed in the ranking. For the steepest examination function, the relative advantage of the \swname\ over the baselines is largest, and it almost doubles the number of matches that is achieved by the \reciprocalname.
\paragraph{How does the size of the markets affect the methods?} 
In this experiment, we vary the size of market to understand how this affects the effectiveness of the methods. Results are shown in the rightmost plot of Figure~\ref{fig:combine}. As market size increases, all the ranking methods achieve higher utility, which is expected since there are more opportunities for matches. More interestingly, regardless of the market size, the relative performance among the methods is largely unaffected.

\subsection{Validation on Real-World Data} 
\label{sec:realworld}

We also validated our method on data from two real-world two-sided matching markets. First, we collected a new dataset by launching a networking recommendation system for a major computer science conference. The goal was to help participants find other participants that they may want to interact with. 
This recommendation system fits naturally in the two-sided matching markets framework as each user acts proactively by sending messages, scheduling meetings, etc to other users and the recommendation is successful if the other user replies positively. To account for the fact that each user can serve as both the proactive side (initiate interaction) and the reactive side (reply to the messages), we put each of the 925 users on both sides of the market, with the two-sided relevances as detailed in the Appendix. 
We tested the performance of various ranking algorithms for this data using $\examfunc{x}=1/x$ as the examination function (results for other examination functions follow a similar trend). 
To reduce computational complexity, we use a two-stage ranking procedure. We first identify the top 100 results based on their reciprocal relevance, and only re-rank those to maximize social welfare. The ranking after the top 100 is by reciprocal relevance. 
Results over 10 runs are shown in Table~\ref{table:realworld}. The substantial improvement over baseline algorithms verifies that the proposed approach can provide significant benefit in realistic applications. Beyond this overall improvement in social welfare, 
we find increased individual utility for more than 88\% of the participants, which we further discuss in the following section on fairness. 

As a second real-world benchmark, we tested our method on data from the online dating service Libimseti \cite{10.1145/2808797.2809282}. 
The dataset contains ratings given by a user to other users in the system. 
We select 500 males and 500 females that have given the most ratings to other users and impute any missing rating using the alternating least squares (ALS) procedure \cite{Paterek2007ImprovingRS}, making the simplifying assumption that females will only rate males and vice versa. Again, we evaluate the performance of various algorithms using $\examfunc{x}=1/x$, and use the same two-stage ranking procedure for \swname. Results are shown in the third column of Table~\ref{table:realworld}. Again, \swname\ achieves the highest social welfare and exceeds the baseline algorithms by a large margin.
\begin{table}[h!]
\caption{Social Welfare ($\pm$ two stderr) for various algorithms on real-world datasets.}
\centering
\begin{tabular}{ ccc } 
\hline
Dataset & Networking Recommendation  & Online Dating (Libimseti)    \\ 
$\naivename$ & 604.0 $\pm$ 0.11 & 844.0 $\pm$ 0.10 \\ 
$\reciprocalname$ & 763.8 $\pm$ 0.06 & 957.2 $\pm$ 0.12 \\ 
$\swname$ & 824.1 $\pm$ 0.18 & 1199.2 $\pm$ 0.14 \\ 
\hline
\end{tabular}
\label{table:realworld}
\end{table}
\subsection{Impact on Individual Fairness, Adoption and Retention} 
\label{sec:fair}
The issue of fairness is more complex in matching markets than in conventional applications of ranking systems, given the complex dependencies between the individuals. While our objective of maximizing social welfare recognizes the societal importance and impact of many matching markets, we also examined its impact on individual fairness. Figure~\ref{fig:fairness} plots the histograms of individual utilities for different ranking algorithms on the synthetic dataset (results for the two real-world datasets are in Appendix~\ref{sec:tradeoff}). On all datasets, \swname\ improves individual utility disparity in comparison to \naivename\ and \reciprocalname. Specifically, the fraction of users in the lowest utility bin is reduced by \swname, leading to a more equitable distribution of utility.
However, we argue that adding explicit fairness constraints is still useful, and that our optimization-based framework is well-suited for including statistical parity or merit-based exposure constraints \cite{singh2018fairness} in future work.
Beyond fairness, we also briefly discuss the need for future work on strategic behavior, such as adoption and retention in Appendix~\ref{sec:tradeoff}.
\begin{figure*}
    \centering
    \includegraphics[width=\textwidth]{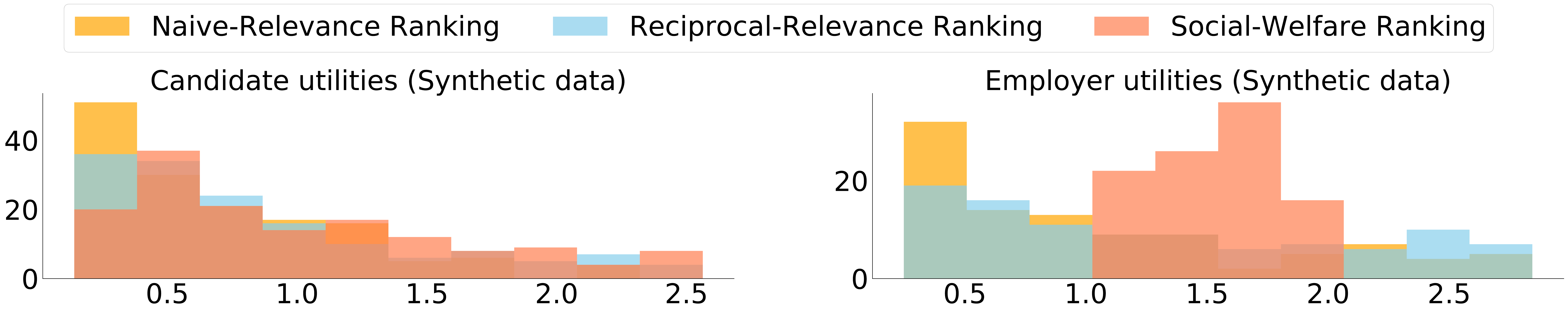}
    \caption{Histogram of candidate and employer utilities on synthetic dataset.}
    \label{fig:fairness}
\end{figure*}
\section{Conclusions}
We have formulated the problem of ranking in the two-sided matching markets with the objective of maximizing social welfare in terms of the total number of matches. To make this problem tractable, we identified a lower bound and show how it can be optimized effectively. This gives the first ranking algorithm that is able to jointly optimize personalized rankings that take into account mutual preferences and limited capacity.
In experiments on both synthetic and real-world datasets, we find that our proposed ranking algorithm can consistently achieve the highest social welfare in comparison to the baselines.
We also outlined directions for future work on strategic behavior and fairness guarantees for ranking systems in matching markets. In addition, this work opens other interesting questions for future work, such as how estimation errors of relevance probabilities affect the rankings and their overall social welfare.

\section*{Acknowledgements}
This research was supported in part by NSF Awards IIS-1901168 and IIS-2008139, as well as a Bloomberg Fellowship. All content represents the opinion of the authors, which is not necessarily shared or endorsed by their respective employers and/or sponsors.


\newpage
\appendix

\section{Appendix}
\subsection{Discussion on Adoption and Retention}
\label{sec:tradeoff}
In this section, we want to at least briefly discuss the need for future work on strategic behavior and fairness guarantees for rankings in matching markets. We posit that these issues are even more important than for ranking systems in conventional markets, since for matching markets the rankings and the actions of all participants shape the utilities of each individual in complex and interdependent ways. This is true for any ranking system applied to matching markets, whether it optimizes social welfare, any other objective, or does not do any explicit optimization at all.

To formalize strategic behavior, we view each candidate $c\in\canset$ in the market as an independent player that can choose among ranking policies $\pi$ (e.g., different recommender systems) or equivalently their doubly stochastic matrices $\dsmatrix{\pi}{c}$. For each player $c$, the payoff $R_c(\dsmatrix{\pi}{c}, \bigcup_{\ccomplement \in\canset\backslash c} \dsmatrix{\pi'}{\ccomplement})$  is the tractable lower bound of total expected number of matches that player $c$ could get, which is a function of player $c$'s action $\dsmatrix{\pi}{c}$, and the other players' actions $\bigcup_{\ccomplement \in\canset\backslash c} \dsmatrix{\pi'}{\ccomplement}$.
\begin{equation*}
	R_c(M^{\pi}_c, \!\!\!\bigcup_{\ccomplement\in\canset\backslash c}\!\!\! M^{\pi'}_{\ccomplement})= \sum_{j\in\jobset} \canpre{c}{j}\jobpre{j}{c}v\big(1+\!\!\!\!\!\!\!\sum_{c' \in \priorset{j}{c}}\!\!\!\!\!\!\canpre{c'}{j}\basisvec{j}^T \dsmatrix{\pi'}{c'} \examvec\big)\basisvec{j}^T \dsmatrix{\pi}{c} \examvec
\end{equation*}

\paragraph{Adoption.}
The first question we consider is whether candidates will want to participate in the social-welfare optimal system, or whether they will prefer the naive ranking $\dsmatrix{\naivepolicy}{c}$ which ranks by their own relevance probabilities $\canpre{c}{j}$? As strategic agents, candidates will switch to their social-welfare optimal ranking $\dsmatrix{\pi}{c}$, if it increases their utility compared to $\dsmatrix{\naivepolicy}{c}$. If we assume that during the initial fielding of the system all candidates $c'\in \canset$ use $\dsmatrix{\naivepolicy}{c'}$, we may want to add the following constraints to our optimization objective, enforcing that all agents have an $\epsilon$-incentive to switch.
\begin{equation}
\label{eq:adopt}
	R_c(\dsmatrix{\pi}{c}, \bigcup_{\ccomplement \in\canset\backslash c} \dsmatrix{\naivepolicy}{\ccomplement}) \geq R_c(\dsmatrix{\naivepolicy}{c}, \bigcup_{\ccomplement \in\canset\backslash c} \dsmatrix{\naivepolicy}{\ccomplement})+\epsilon 
\end{equation}
It is worth noting that these constraints are linear in $\dsmatrix{\pi}{c}$. Hence, we can simply incorporate them into the Frank-Wolfe Algorithm (see Algorithm~\ref{alg:1}). 

\paragraph{Retention.}
The second question considers the behavior of the candidates once the system has been widely adopted. In particular, do candidates have an incentive to abandon the system and return to their naive ranking $\dsmatrix{\naivepolicy}{c}$?

To avoid this, the system ranking $\dsmatrix{\pi}{c}$ should provide a larger utility than $\dsmatrix{\naivepolicy}{c}$ given that all other candidates stay with the system, which could be enforced by adding additional constraints of the form

\begin{equation}
	R_c(\dsmatrix{\pi}{c}, \bigcup_{\ccomplement \in\canset\backslash c} \dsmatrix{\pi}{\ccomplement})\geq R_c(\dsmatrix{\naivepolicy}{c}, \bigcup_{\ccomplement \in\canset\backslash c} \dsmatrix{\pi}{\ccomplement})+ \epsilon.
\end{equation}

Unfortunately, this constraint set is not convex and it is difficult to use Algorithm~\ref{alg:1} (similarly Algorithm~\ref{alg:3}) directly. However gradient descent ascent (GDA) could be used to find the solution of the Lagrangian dual form of the problem.


\paragraph{Experimental Examination.} 
In this experiment, we examine how does the social-welfare optimized ranking affect individual utilities.
As discussed in aforementioned paragraphs, the individual utilities of the candidates can affect adoption, retention and fairness of the ranking system. The first row of Figure~\ref{fig:standard_stable_adopt} considers the case where all users switch from using the \naivename\ to using the \swname\ on the synthetic dataset. It shows a histogram of candidates according to how much they gain from this switch in terms of expected number of matches. Surprisingly, none of the candidates is worse off in this switch for the synthetic dataset. Nevertheless, there are still potential fairness issues and we see on the real-world datasets that this uniformity in gain is not always guaranteed. The second row considers the gain from adoption, which is also non-negative for all candidates on the synthetic dataset and a large fractional of candidates on the networking recommendation dataset. However, the proportion of positive gains on the online dating dataset is pretty low. This calls for interesting future work on incorporating adoption in recommender system design.
Finally, the third row shows that all candidates in the synthetic dataset (or most candidates for real-world datasets) are better off staying in the system than switching back to the \naivename. We also explored other variants of our synthetic data, and generally found that the \swname\ is beneficial for most users. 
This is pretty encouraging for future work, since it suggests that stronger guarantees on individual utilities may be achievable with little drop in social welfare. 

Beyond strategic behaviors such as adoption and retention, we also empirically examine the fairness in individual utility gains.  Figure~\ref{fig:fairness_real} shows the histograms of the candidate/employer utilities for two real-world datasets. Similar as what we observed for synthetic dataset, \swname\ leads to a more equitable distribution of utility, compared with the two baseline algorithms. 
A full treatment of strategic behavior and fairness for ranking systems in matching markets is beyond the scope of this paper, but is an important area for future work.

\begin{figure}[t!]
    \centering
    \includegraphics[width=0.32\linewidth]{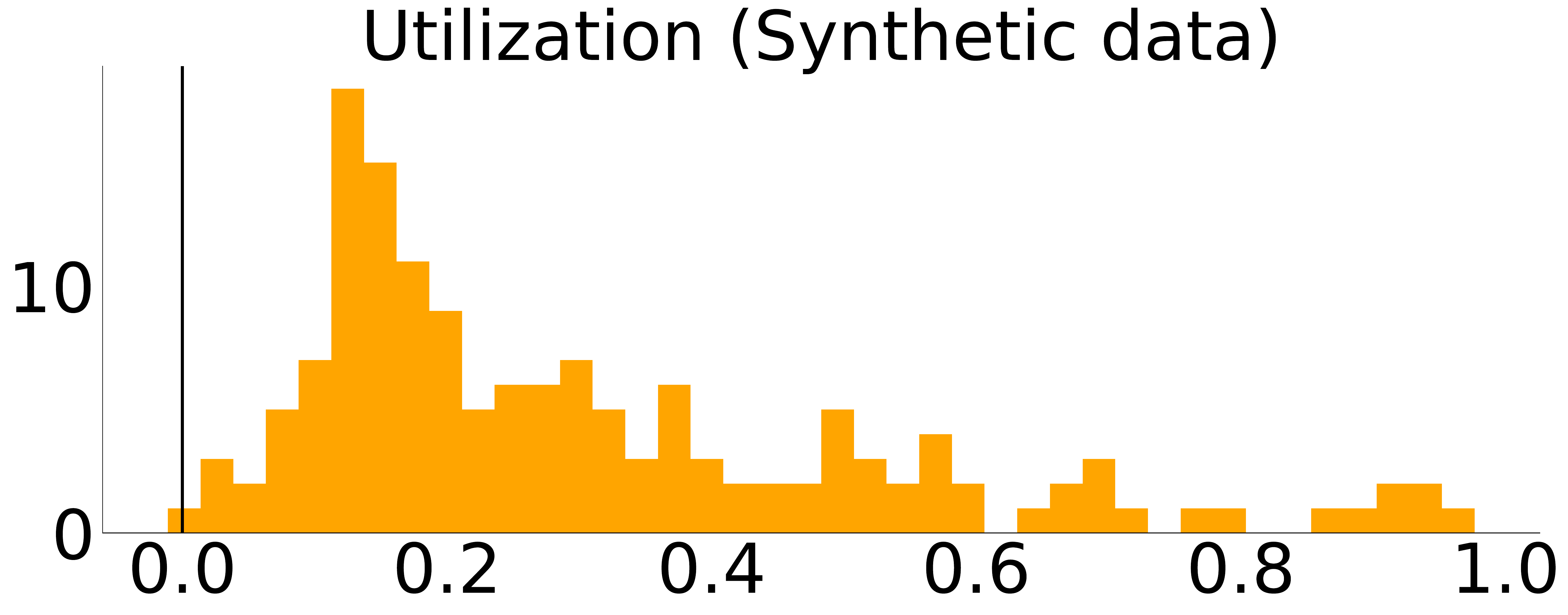}
    \includegraphics[width=0.32\linewidth]{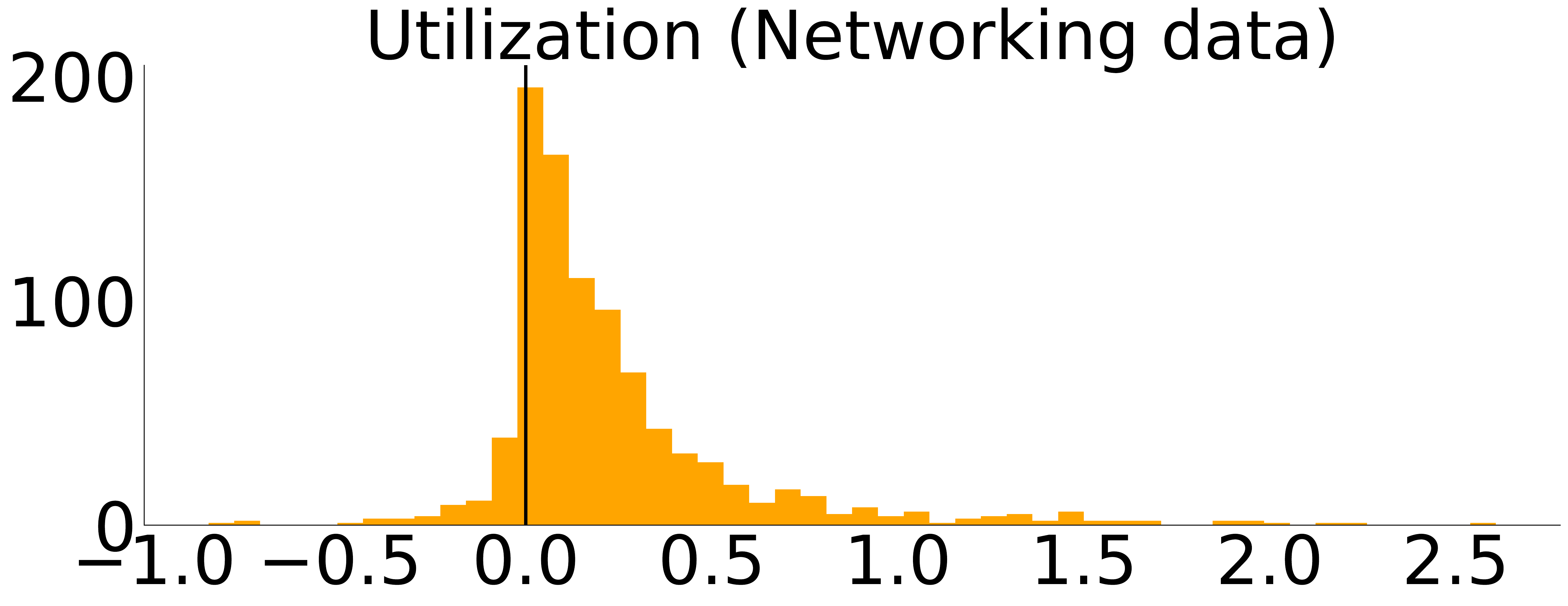}
    \includegraphics[width=0.32\linewidth]{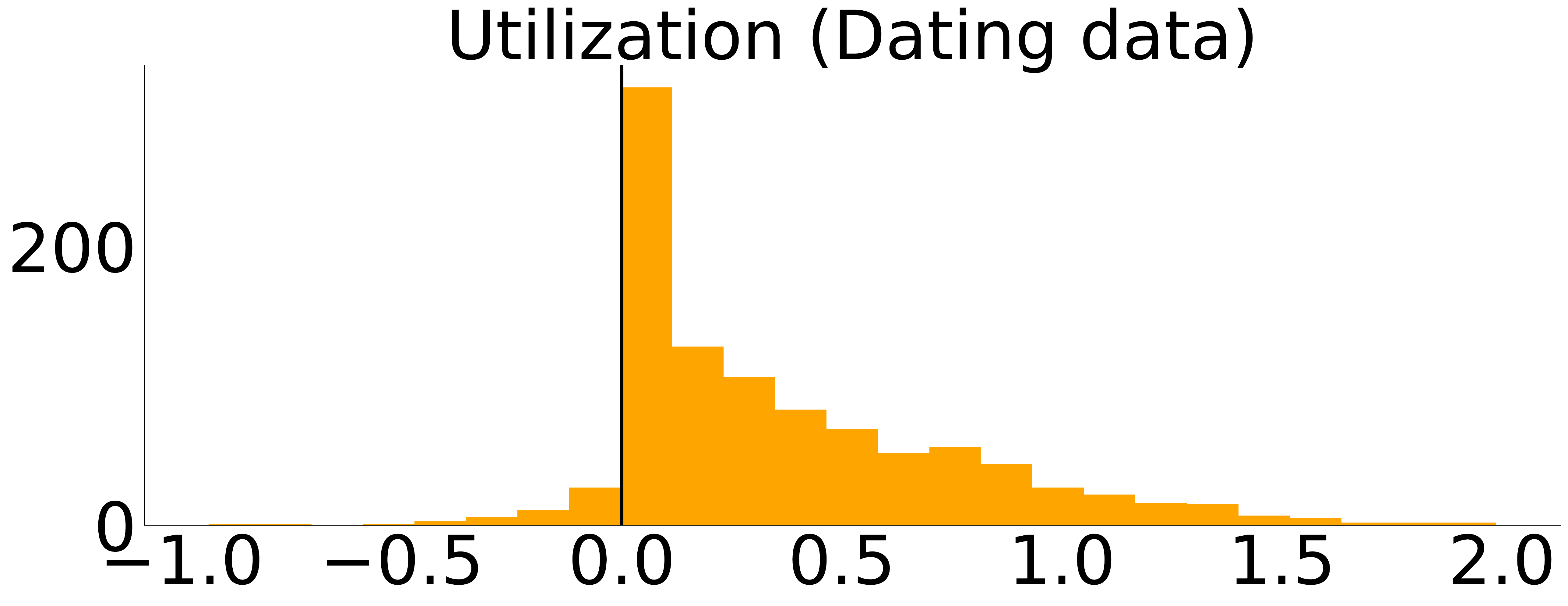}
    
    \includegraphics[width=0.32\linewidth]{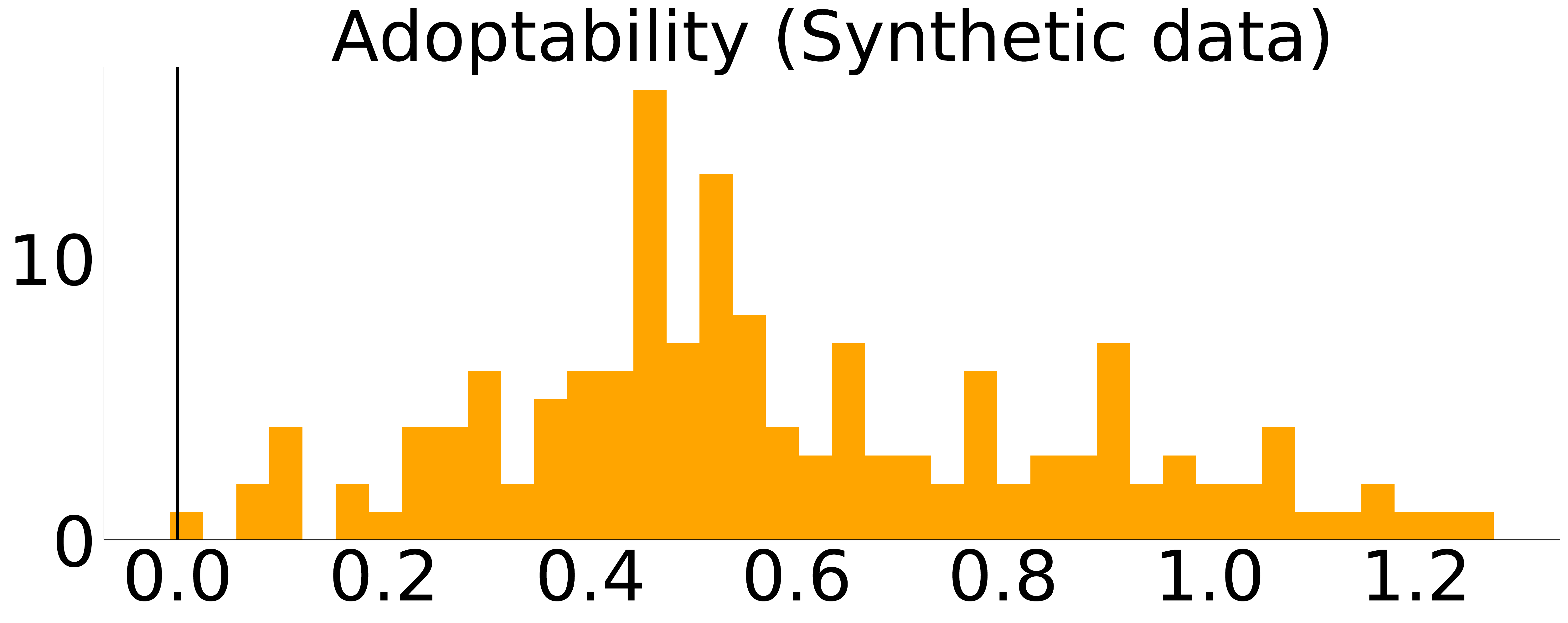}
     \includegraphics[width=0.32\linewidth]{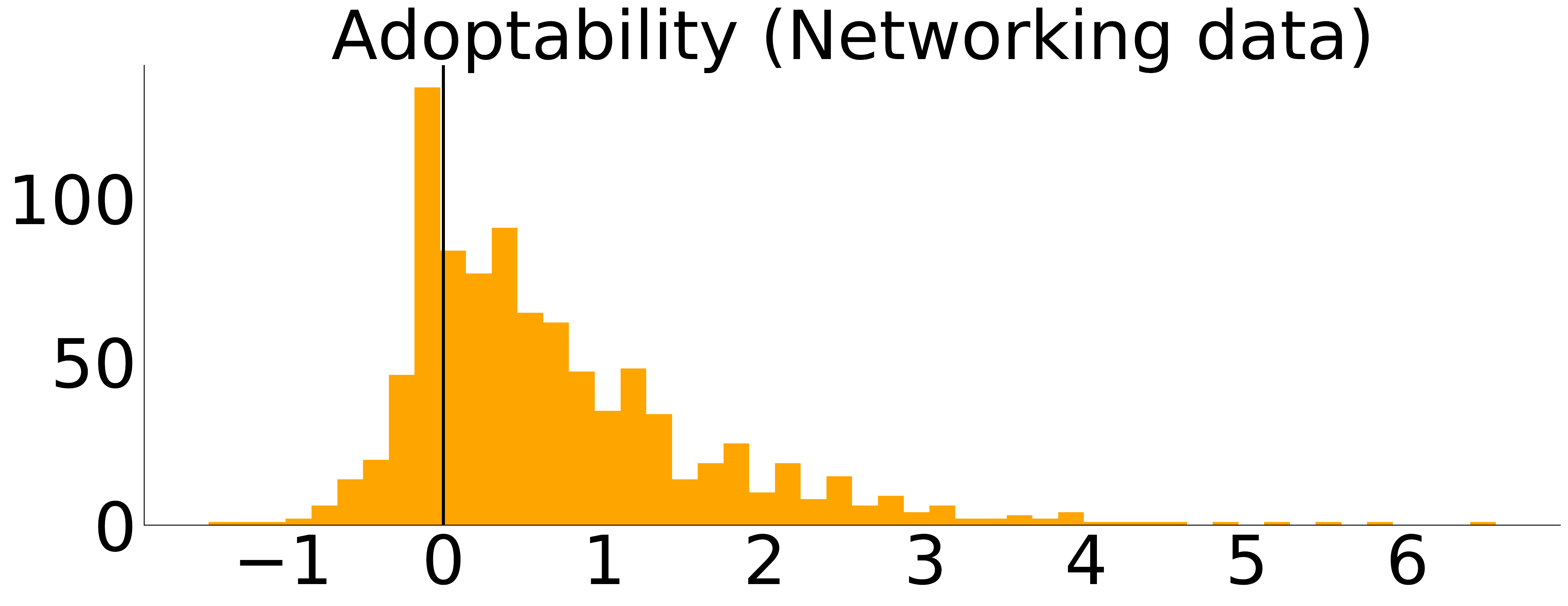}
     \includegraphics[width=0.32\linewidth]{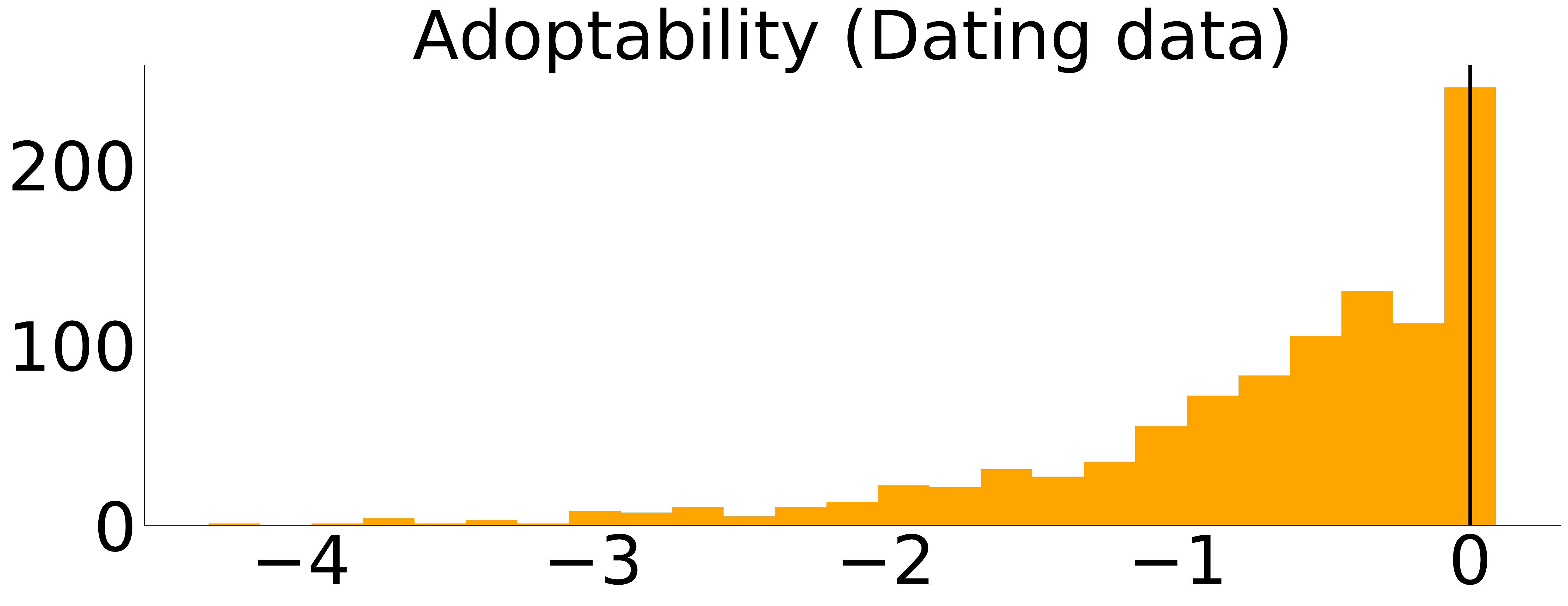}
     
    \includegraphics[width=0.32\linewidth]{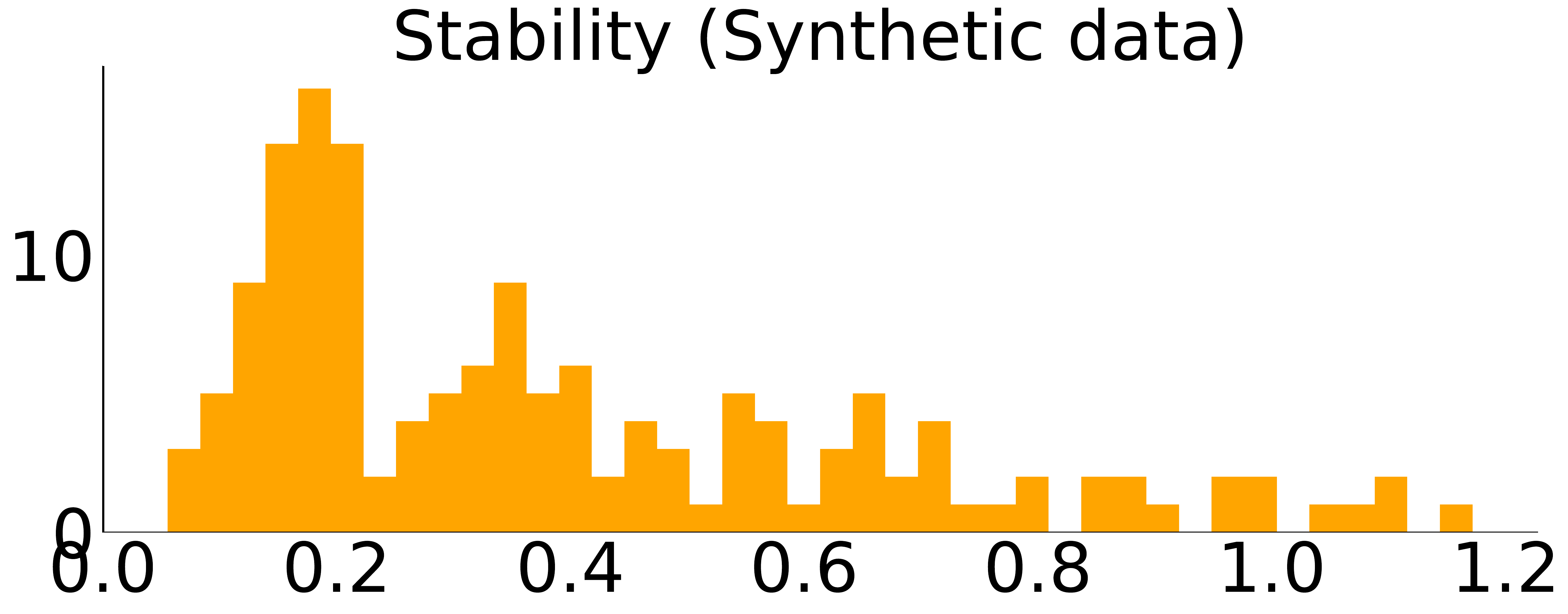}
    \includegraphics[width=0.32\linewidth]{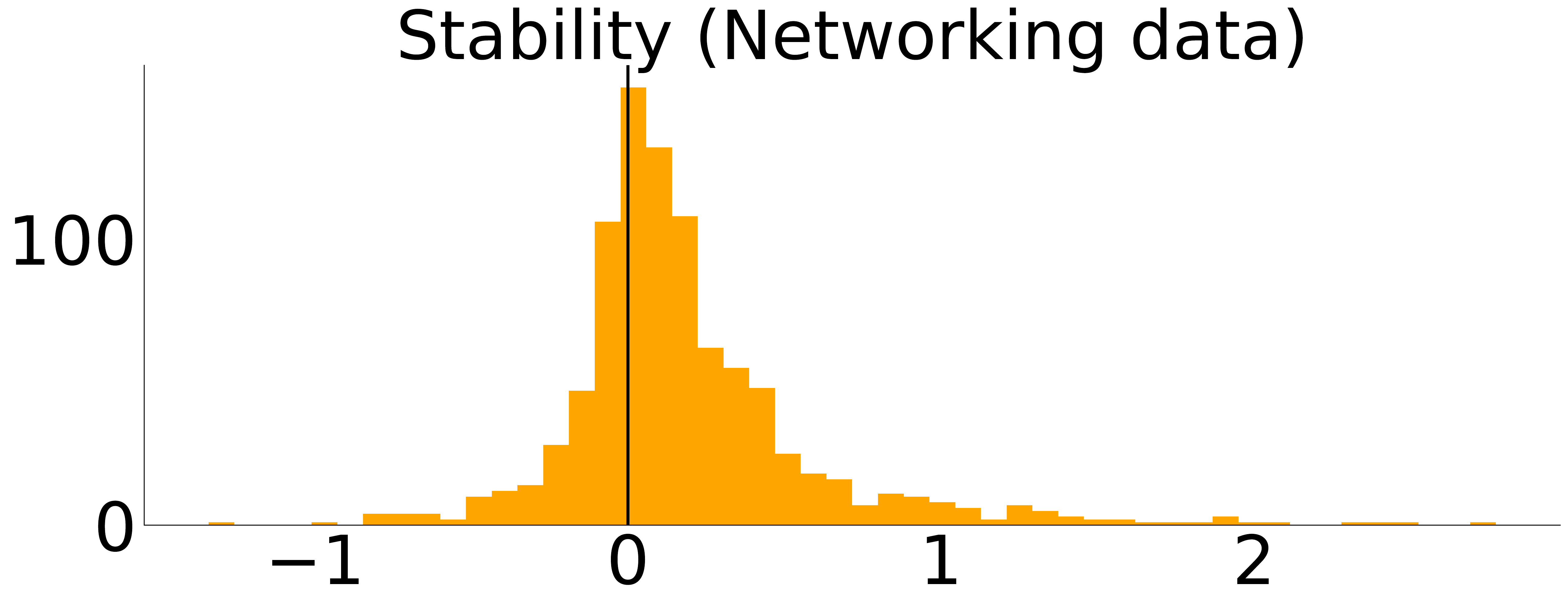}
    \includegraphics[width=0.32\linewidth]{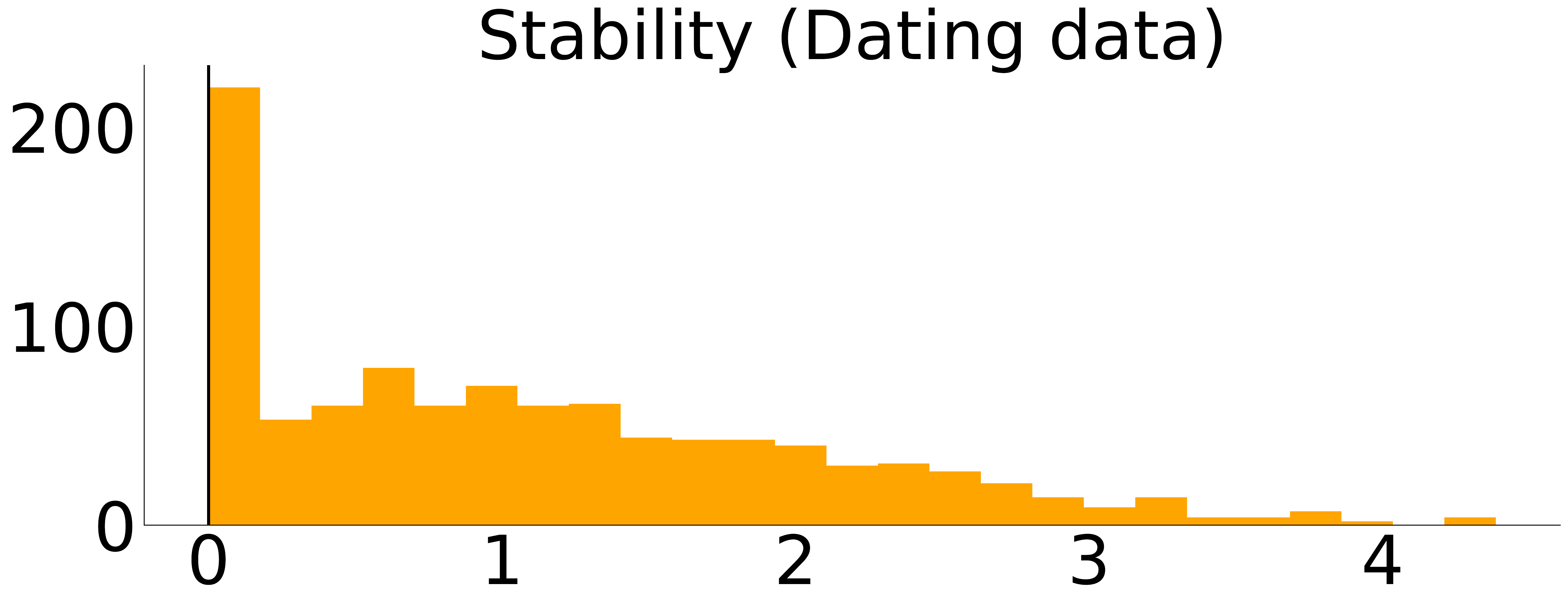}

    \caption{Histogram of individual utility gains on synthetic data with $\lambda=0.5$, random, $\examfunc{x}=1/x$ and $n=100$ (left column) and two real-world datasets: Networking Recommendation (middle column) and Online Dating (right column). (Top) Gain when all users switch from \naivename\ to \swname. (Middle) Gain from adopting \swname\ even though others are not. (Bottom) Gain from not abandoning \swname\ if all others keep using it. }
    \label{fig:standard_stable_adopt}
\end{figure}
    
\begin{figure*}
    \centering
    \includegraphics[width=\textwidth]{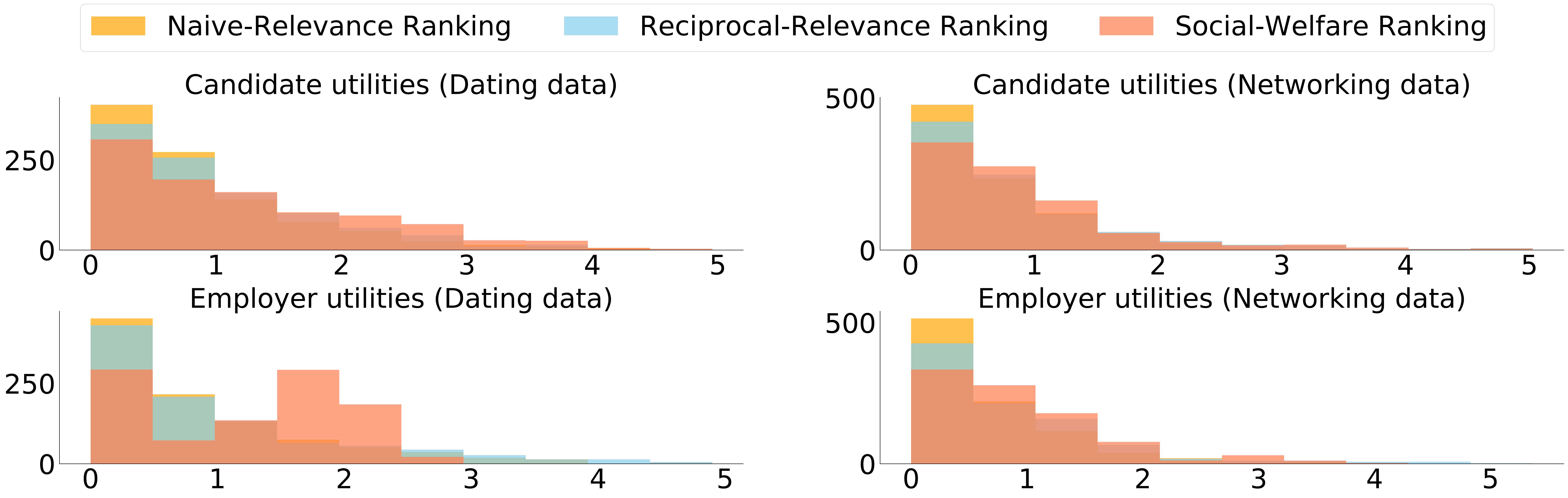}
    \caption{Histogram of candidate and employer utilities on two real-world datasets.}
    \label{fig:fairness_real}
\end{figure*}

\subsection{Discussion on Comparison with Stable Matching}
\label{app:stable}
Our ranking problem is fundamentally different from stable matching problems considered in the economics literature \cite{gale1962college, masarani1989existence}. First, our ranking problem aims at assisting users in forming their preferences over the options by focusing the user's attention, in contrast to assuming that users can readily state their preferences for the matching procedure. 
However, it still leaves open the question whether stable matching procedures could be used to solve our optimization problem. 

To address this question, consider the following simplified \emph{top-1 ranking problem} where all participants only examine position one of their ranking (i.e., $\examfunc{1}=1$ and $\examfunc{k}=0$ for all $k$ and for both sides). Furthermore, we assume that the relevance probabilities $\canpre{c}{j}$ and $\jobpre{j}{c}$ imply a preference order for each candidate and employer. The following proposition clarifies that even in this simplified setting stable matching and social-welfare maximization are not equivalent.
\begin{restatable}{proposition}{comparemm}
\label{compare_mm}
	A matching that is stable is not necessarily a social-welfare optimal top-1 ranking, and top-1 rankings that are not stable matchings can have better social welfare than matchings that are stable.
\end{restatable}
\begin{proof}
We prove the claim by considering the two-sided market shown below. There are three job candidates $\{c_1, c_2, c_3\}$ and three employers $\{j_1, j_2, j_3\}$. The candidates' relevance table is shown at left, where each row is a ranked list of employers sorted by the candidate-side relevance probability $\canpre{c}{j}$. Similarly, the table on the right is the ranked list of candidates sorted by the employer-side relevance probability $\jobpre{j}{c}$. 
\begin{center}
\begin{tabular}{ cc } 
User Relevance Table  & Employer Relevance Table\\
\begin{tabular}{ |c |c  c  c|}
\hline
 $c_1$ & $j_1$ & $j_3$ & $j_2$ \\ 
 \hline
 $c_2$ & $j_2$ & $j_1$ & $j_3$ \\  
 \hline
 $c_3$ & $j_1$ & $j_2$ & $j_3$\\  
 \hline
 $\canpre{c}{j}$ & $1$ & $0.9$ & $0.1$\\  
 \hline
\end{tabular} &
\begin{tabular}{ |c |c  c  c|}
\hline
 $j_1$ & $c_1$ & $c_3$ & $c_2$ \\ 
 \hline
 $j_2$ & $c_2$ & $c_1$ & $c_3$ \\  
 \hline
 $j_3$ & $c_1$ & $c_2$ & $c_3$\\  
 \hline
 $\jobpre{j}{c}$ & $1$ & $0.9$ & $0.1$\\  
 \hline
\end{tabular}\\
\end{tabular}
\end{center}
The matching $\{(c_1, j_1), (c_2,j_2), (c_3,j_3)\}$ is stable, and it provides a social welfare of $1+1+0.1\times 0.1=2.01$. Now consider the ranking policy $\pi$ so that $\pi(\cdot|c_1)$ recommends $j_3$ at rank 1 with probability 1, $\pi(\cdot|c_2)$ recommends $j_2$ at rank 1 with probability 1, and $\pi(\cdot|c_3)$ recommends $j_1$ at rank 1 with probability 1. This ranking policy provides a social welfare of $1 + 0.9+0.9 = 2.8$, but it is not a stable matching ($c_1$ and $j_1$ prefer each other over their current match).
\end{proof}

\subsection{Proof}
\label{app: proof}
In this section, we provide the proof of lemmas and theorems stated in the main paper.

\setcounter{theorem}{2}

\subsubsection{Proof of Theorem 2}
\subopt*
\begin{proof}
In this theorem we show that the gap of the naive relevance-based ranking algorithm and the optimal solution by our algorithm (Equation~\ref{eq:lower_opt}) can be $\Theta(n)$, and it is achieved by identifying some specific instance of the two-sided matching market (particular choice of relevance probability and examination model). Consider the following instance of a fixed two-sided market: there are $n$ employers (denoted as $\{j_1, j_2, \cdots, j_n\}$) and $n$ candidates (denoted as $\{c_1, c_2, \cdots, c_n\}$) in the market, with $n\geq 2$. On the candidate side, they have highly correlated relevance probability over jobs, while each job has different relevance probability over candidates. The employers' relevance probability ordering (ordered by descending order of $\jobpre{j}{\cdot}$) is given by the following circulant matrix:
\begin{center}
\begin{tabular}{ c |c  c  c c c}
 $j_1$ & $c_1$ & $c_2$ & $c_3$ & $\cdots$ & $c_n$  \\ 
 \hline
 $j_2$ & $c_2$ & $c_3$ & $\cdots$ & $c_n$ & $c_1$\\  
 \hline
 $j_3$ & $c_3$ & $\cdots$ & $c_n$ & $c_1$ & $c_2$\\  
 \hline
 \vdots &  \vdots &\vdots &\vdots&\vdots&\vdots  \\
 $j_k$ & $c_k$  & $\cdots$& $\cdots$& $c_{k-2}$ &$c_{k-1}$ \\
  \vdots &  \vdots &\vdots&\vdots&\vdots&\vdots \\
 \hline
 $j_n$ &   $c_n$ & $c_1$ & $\cdots$ & $c_{n-2}$ & $c_{n-1}$
\end{tabular}
\end{center}
while the top 3 employers of the candidates' relevance probability list (ordered by descending order of $\canpre{c}{\cdot}$) is:
\begin{center}
\begin{tabular}{ c |c  c  c c c}
 $c_1$ & $j_1$ & $j_2$ & $j_3$  \\ 
 \hline
 $c_2$ &$j_1$ & $j_2$ & $j_3$ \\  
 \hline
 $c_3$ & $j_1$ & $j_3$ & $j_2$  \\  
 \hline
 \vdots &  \vdots &  \vdots & \vdots  \\
 \hline
 $c_k$ &  $j_1$ & $j_k$  & $j_{k-1}$ \\
 \hline
 \vdots &  \vdots&  \vdots &  \vdots  \\
 \hline
 $c_n$ &   $j_1$ & $j_n$ & $j_{n-1}$
\end{tabular}
\end{center}
and for each candidate $c_i$, for the employers $j$ that not in the top 3 position of $c_i$'s relevance list, we set arbitrary relevance probability $\canpre{c_i}{j}$ with $\canpre{c_i}{j}\in [0,\canpre{c_i}{\caninrank{c_i}{3}}]$. 

Given this, we define the cardinal relevance probability $\canpre{c}{\cdot}$ and $\jobpre{j}{\cdot}$ from the above ordinal ordering as:
$$\canpre{c_i}{j_k} = \frac{n-(\rankinlist{c_i}{j_k}-1)}{n} \hspace{1cm}
\jobpre{j_k}{c_i} = \frac{n-(\rankinlist{j_k}{c_i}-1)}{n}$$
The examination model $E_m$ in this case is given by: 
\[
    \examfunc{x}= 
\begin{cases}
    0.1^{x-1},& \text{if } x\leq m\\
    0,              & \text{otherwise}
\end{cases}
\]
This mimics the scenario that even as the market size $n$ grows, people tends to only examine the top $m$ recommendations due to time and resource constraints. Here we choose $m=2$ for simplicity in the proof, but it could be generalize to any fixed $m$. As candidates only examine the top 2 positions, then only the top 2 positions in our rankings do matter. Therefore for all rankings we compare in the proof, we will only list the top 2 position in the ranking. The naive one-side relevance based ranking $\naivepolicy$ in this case corresponds to ignoring the two-sided perspective, and only recommending the employer to candidates by how relevant they are in the candidate's perspective. Therefore $\naivepolicy$ is a deterministic policy with probability 1 on the following permutation $\sigma^n$:
\begin{equation}
\begin{aligned}
       &\rank{j_1}{\sigma^n(c_1)} = 1 \hspace{1cm} \rank{j_2}{\sigma^n(c_1)}  =2 \\
    &\rank{j_1}{\sigma^n(c_i)} = 1 \hspace{1cm} \rank{j_i}{\sigma^n(c_i)}  =2 \hspace {0.2cm}\forall i\geq 2         
\end{aligned}
\end{equation}
Now we analyze the utility $\swobj{\naivepolicy}$, it is worth noting that the employers will also only examine the top 2 positions. Then we have
\begin{equation}
\begin{aligned}
        \indiswobj{\naivepolicy}{c_1} &= \PP(\applied{\naivepolicy}{c_1}{j_1}) + \PP(\applied{\naivepolicy}{c_1}{j_2})\\
        &=1+\frac{0.1(n-1)}{n}\times \frac{1}{n}\bigg(\PP(\rankinapplied{\naivepolicy}{j_2}{c_1}=1)+0.1\PP(\rankinapplied{\naivepolicy}{j_2}{c_1}=2)\bigg)\\
         &=1+\frac{0.1(n-1)}{n}\times \frac{1}{n}\bigg(1-\frac{0.1(n-1)}{n}+0.1\frac{0.1(n-1)}{n} \bigg)\\
        & \leq 1.1+\frac{0.9}{n}
\end{aligned}
\end{equation}
For candidate $c_2$, we have
\begin{equation}
\begin{aligned}
        \indiswobj{\naivepolicy}{c_2} &= \PP(\applied{\naivepolicy}{c_2}{j_1}) + \PP(\applied{\naivepolicy}{c_2}{j_2})\\
        &=\frac{0.1(n-1)}{n}+\frac{0.1(n-1)}{n}= \frac{0.2(n-1)}{n}
\end{aligned}
\end{equation}
For candidate $c_i$ with $i\geq 3$, we have
\begin{equation}
\begin{aligned}
        \indiswobj{\naivepolicy}{c_i} &= \PP(\applied{\naivepolicy}{c_i}{j_1}) + \PP(\applied{\naivepolicy}{c_i}{j_i})\\
        &=0+\frac{0.1(n-1)}{n}= \frac{0.1(n-1)}{n}
\end{aligned}
\end{equation}
Then the social welfare generated by policy $\naivepolicy$ is:
\begin{equation}
\begin{aligned}
            \swobj{\naivepolicy} &\leq 1.1+\frac{0.9}{n}+\frac{0.2(n-1)}{n}+\frac{0.1(n-1)(n-2)}{n}\\
            &\leq 1+0.1n+\frac{1}{n}
\end{aligned}
\end{equation}
Now we consider a different policy $\pi^s$ which is a deterministic policy with probability 1 on the following permutation $\sigma^s$:
\begin{equation}
\begin{aligned}
       &\rank{j_1}{\sigma^s(c_1)} = 1 \hspace{1cm} \rank{j_2}{\sigma^s(c_1)}  =2 \\
    &\rank{j_i}{\sigma^s(c_i)} = 1 \hspace{1cm} \rank{j_{i-1}}{\sigma^s(c_i)}  =2 \hspace {0.2cm}\forall i\geq 2\\
    \end{aligned}
\end{equation}
Similarly we analyze the utility $\swobj{\pi^s}$ and we have:
\begin{equation}
\begin{aligned}
        \indiswobj{\pi^s}{c_1} &= \PP(\applied{\pi^s}{c_1}{j_1}) + \PP(\applied{\pi^s}{c_1}{j_2})\\
        &=1+\frac{0.1(n-1)}{n}\times \frac{1}{n}\bigg(\PP(\rankinapplied{\pi^s}{j_2}{c_1}=1)+0.1\PP(\rankinapplied{\pi^s}{j_2}{c_1}=2)\bigg)\\
         &=1+\frac{0.1(n-1)}{n}\times \frac{1}{n}\bigg(1-\frac{(n-1)}{n}+0.1\frac{(n-1)}{n} \bigg)\\
        & \geq 1+\frac{0.01(n-1)}{n^2}
\end{aligned}
\end{equation}
For candidate $c_2$, we have
\begin{equation}
\begin{aligned}
        \indiswobj{\pi^s}{c_2} &= \PP(\applied{\pi^s}{c_2}{j_2}) + \PP(\applied{\pi^s}{c_2}{j_1})\\
        &=\frac{(n-1)}{n}+\frac{0.1^2(n-1)}{n}= \frac{1.01(n-1)}{n}
\end{aligned}
\end{equation}
For candidate $c_i$ with $i\geq 3$, we have
\begin{equation}
\begin{aligned}
        \indiswobj{\pi^s}{c_i} &= \PP(\applied{\pi^s}{c_i}{j_i}) + \PP(\applied{\pi^s}{c_i}{j_{i-1}})\\
        &=\frac{(n-1)}{n}+\frac{0.1(n-2)}{n}\bigg(\PP(\rankinapplied{\pi^s}{j_{i-1}}{c_i}=1)+0.1\PP(\rankinapplied{\pi^s}{j_{i-1}}{c_i}=2)\bigg)\\
        &=\frac{(n-1)}{n}+\frac{0.1(n-2)}{n}\bigg(1-\frac{n-1}{n}+\frac{0.1(n-1)}{n}\bigg)\\
        &\geq \frac{n-1}{n}
\end{aligned}
\end{equation}
Then the social welfare generated by policy $\pi^s$ is:
\begin{equation}
\begin{aligned}
            \swobj{\pi^s} &\geq 1+\frac{0.01(n-1)}{n^2}+ \frac{1.01(n-1)}{n} +\frac{(n-1)(n-2)}{n}\\
            &\geq n-0.99+\frac{0.99}{n}
\end{aligned}
\end{equation}
Now we denote the optimal solution for social welfare objective as $\optimalpolicy$ and we have $\swobj{\optimalpolicy} - \swobj{\naivepolicy}\geq \swobj{\pi^s} - \swobj{\naivepolicy} = \Theta(n)$.
\end{proof}

\subsubsection{Proof of Lemma 3}
\pdf*
\begin{proof}
Fix any stochastic ranking policy $\pi$, for employer $j$ and candidate $c'$, we use $\applied{\pi}{c'}{j}$ to denote the event that candidate $c'$ applies to employer $j$, and it is easy to see that $\applied{\pi}{c'}{j}$ is a Bernoulli RV with parameter $\PP(\applied{\pi}{c'}{j}=1) = \clickpro{\pi}{c'}{j}$. Also, the rank of candidate $c$ in employer $j$'s list is $1+$ number of candidates who has higher relevance probability to $j$ (i.e., candidates in the set $\priorset{j}{c}$) and also applied to employer $j$.
\begin{equation}
	\rankinapplied{\pi}{j}{c} \stackrel{\Dcal}{\sim} 1 + \sum_{l=1}^{\rankinlist{j}{c}-1} \applied{\pi}{\caninrank{j}{l}}{j}
\end{equation}
Then it is easy to see that the sum of independent Bernoulli trails (but not necessarily identically distributed) follows Poisson Binomial with probabilities $ \big[\clickpro{\pi}{\caninrank{j}{1}}{j}, \clickpro{\pi}{\caninrank{j}{2}}{j}, \cdots, \clickpro{\pi}{\caninrank{j}{\rankinlist{j}{c}-1}}{j}\big]$, and the PMF could be easily shown as in Equation~\ref{eq:dist}.
\end{proof}

\subsubsection{Proof of Theorem 4}
\uti*
\begin{proof}
The proof is based on the convexity of $v$ and the simple form of the expectation of the Poisson Binomial random variable.
\begin{equation}
\begin{aligned}
\label{SW_lower}
\swobj{\pi}
	&= \sum_{c\in\canset}\sum_{j\in\jobset} \big(\canpre{c}{j}\basisvec{j}^T\dsmatrix{\pi}{c} \examvec\big)\jobpre{j}{c}\EE\big[v\big(\rankinapplied{\pi}{j}{c}\big)\big]\\
	&\geq \sum_{c\in\canset}\sum_{j\in\jobset} \big(\canpre{c}{j}\basisvec{j}^T\dsmatrix{\pi}{c} \examvec\big)\jobpre{j}{c}v\big(\EE\big[\rankinapplied{\pi}{j}{c}\big]\big)\\
	&=\sum_{c\in\canset}\sum_{j\in\jobset} \canpre{c}{j}\jobpre{j}{c}v\big(1+\!\!\!\!\!\sum_{c' \in \priorset{j}{c}}\!\!\!\!\!\clickpro{\pi}{c'}{j}\big)\basisvec{j}^T \dsmatrix{\pi}{c} \examvec\\
	&=\sum_{c\in\canset}\sum_{j\in\jobset} \canpre{c}{j}\jobpre{j}{c}v\big(1+\!\!\!\!\!\sum_{c' \in \priorset{j}{c}}\!\!\!\!\!\canpre{c'}{j}\basisvec{j}^T \dsmatrix{\pi}{c'} \examvec\big)\basisvec{j}^T \dsmatrix{\pi}{c} \examvec\\
	&:=\swlowobj{\pi}
\end{aligned}
\end{equation} 
The inequality comes from Jensen's inequality for the convex $\examfunc{\cdot}$. The final equality is based on the expectation of the $\rankinapplied{\pi}{j}{c}$:
\begin{equation}
	\EE[\rankinapplied{\pi}{j}{c}] = 1 + \EE[\sum_{l=1}^{\rankinlist{j}{c}-1} \applied{\pi}{\caninrank{j}{l}}{j}] = 1+\sum_{c' \in \priorset{j}{c}}\clickpro{\pi}{c'}{j}
\end{equation}
\end{proof}

\subsection{Additional Optimization Algorithm}
In this section, we provide the algorithms for optimizing the social welfare lower bound (Equation~\ref{eq:lower_opt}) using either Frank-Wolfe or projected gradient descent.
\begin{algorithm}
\SetAlgoLined
\KwResult{the doubly stochastic matrices, one for each candidate: $\dsmatrices{\pi}=\{\dsmatrix{\pi}{c}\}_{c\in\Ccal}$ }
\textbf{Input:} relevance probabilities $\canpre{c}{\cdot}$ and $\jobpre{j}{\cdot}$, examination function $\examfunc{\cdot}$, stopping criterion $\epsilon$, timesteps $T$, learning rate $\eta_t$\;
 Initialize $\dsmatrix{\pi}{c} = \mathbf{1}\mathbf{1}^T/|\jobset|, \forall c\in\canset$\;
 \For{$t=0,1,\cdots,T$}{
 $S^{*}_{\canset}\in\argmin_{S_{\canset}} -\nabla \underline{\textbf{SW}}(\dsmatrices{\pi})^T S_{\canset}$\;
\hspace{0.7cm} subject to \hspace{0.05cm} $\textbf{1}^T S_{c} = \textbf{1}^T, S_c\textbf{1}=\textbf{1}, \forall c\in\canset$\; \hspace{2.2cm} $S_{\canset}=\{S_c\}_{c\in\canset}$\;
  
 $\dsmatrix{\pi}{\canset} \leftarrow (1-\eta_t)\dsmatrices{\pi} + \eta_t S^{*}_{\canset}$}
 \caption{Social-Welfare Optimization via Frank-Wolfe}
 \label{alg:1}
\end{algorithm}

\label{app:alg}
\begin{algorithm}
\SetAlgoLined
\KwResult{the doubly stochastic matrices: $\dsmatrices{\pi}$ }
\textbf{Input:} preferences $\canpre{c}{\cdot}$, $\jobpre{j}{\cdot}$, examination $\examfunc{\cdot}$ and $\examvec$, stopping criteria $\epsilon$, timestep $T$, learning rate $\eta$\;
 Initialize $\dsmatrix{\pi}{c} = \mathbf{1}\mathbf{1}^T/|\jobset|, \forall c\in\canset$\;
 \For{$t=0,1,\cdots,T$}{
 $\dsmatrix{\pi}{c} \leftarrow \dsmatrix{\pi}{c} + \eta \nabla \underline{\textbf{SW}}$\;
 $D_0=E_0=I$, $k=0$\;
  \While{$\sum_{c}||\textbf{1}^T \dsmatrix{\pi}{c}- \textbf{1}^T||_1\geq \epsilon$ or $\sum_{c}||\dsmatrix{\pi}{c}\textbf{1}- \textbf{1}||_1\geq \epsilon$}{
  $r_k=D_{k-1}\dsmatrix{\pi}{c} E_{k-1}\mathbf{1}$, 
  $D_k = \diag(r^{-1}_k)$\;
  $c^T_k = \mathbf{1}^TD_k \dsmatrix{\pi}{c}E_{k-1}$,
  $E_k = \diag(c^{-1}_k)$\;
  $k \leftarrow k+1$
 }
 $\dsmatrix{\pi}{c} \leftarrow D_{k}\dsmatrix{\pi}{c} E_{k}$}
 \caption{Social-Welfare Optimization via Projected Gradient Descent}
\label{alg:3}
\end{algorithm}


\subsection{Additional Details on Dataset and Implementation} \label{sec:algdetail}
\paragraph{Datasets.} We calculate the two-sided relevances in the networking recommendation dataset as follows.
In the recommendation phrase, the recommendation system estimated relevance scores between all registered users based on features likes similarity of published articles, past co-authorship, past citations, etc. For each user, the system recommended a ranked list of 150 participants, and we collected different forms of positive interactions between them (such as thumb up, send a message, schedule a meeting). In the relevance imputation phrase, we use an importance-weighted (to de-bias position bias) logistic regression with $L_2$ regularization to learn the relevance function, which allows us to impute directional relevance probabilities between all users. Curiously, the relevance probabilities have substantial directionality, with faculty and postdocs being "crowded". However, the expected number of individuals that each user reaches out is rather uniform between groups. For the online dating dataset \cite{10.1145/2808797.2809282} , the license could be accessible at \url{http://web.archive.org/web/20180402034337/http://www.occamslab.com/petricek/data/}.

\paragraph{Implementation.}
For all experiments, the relevance-based ranking is a deterministic ranking based on descending order of $\canpre{c}{j}$, and the reciprocal-based ranking is a deterministic ranking based on descending order of $\canpre{c}{j}\jobpre{j}{c}$. For our algorithm, \swname, we use Algorithm~\ref{alg:1} with 50 timesteps, stopping criterion $10^{-3}$, and constant learning rate $\eta_t=0.2$. We also tested a decaying learning rate $\eta_t=\frac{1}{t+2}$ and the performance was similar to the constant learning rate. 

The convex solver we use for finding the gradient direction within Algorithm~\ref{alg:1} is \texttt{CVXPY}\cite{diamond2016cvxpy} with SCS. All experiments are performed on a Linux computer cluster. To facilitate reproducibility, we also provide the code along with further experiment details. The dataset for networking recommendation will also be published as a benchmark for research in this field.

\bibliography{refs}
\bibliographystyle{plain}
\end{document}